\documentclass[10pt]{article}
\usepackage{newinutile}
\usepackage{cite}

\usepackage{latexsym,amsmath,amsfonts,amsthm,amssymb,bbm,amsbsy}
\usepackage{epsfig,graphics,color,calc,graphicx,pict2e}
\usepackage{comment} 
\usepackage{multirow}
\usepackage[T1]{fontenc}
\usepackage[utf8]{inputenc}
\usepackage{multicol}
\usepackage{float} 

\usepackage{mathrsfs}
\usepackage{tikz}
\usetikzlibrary{shapes,arrows}
\usetikzlibrary{intersections,matrix, positioning}

\usepackage{subfigure} 
\newtheorem{theorem}{Theorem}
\newtheorem{condition}{Condition}
\newtheorem{lemma}{Lemma}
\newtheorem{corollary}{Corollary}
\newtheorem{remark}{Remark}
\newtheorem{proposition}{Proposition}

\def\theequation{\arabic{section}.\arabic{equation}}
\def\thefigure{\arabic{figure}}

\newcommand{\newatop}[2]{\genfrac{}{}{0pt}{}{#1}{#2}}

\newcommand{\cc}[1]{{\mathcal{#1}}}

\newcommand{\puno}{{\pmb{+1}}}
\newcommand{\zero}{{\pmb{0}}}
\newcommand{\muno}{{\pmb{-1}}}
\newcommand{\pmuno}{{\pmb{\pm1}}}

\newlength{\pecettawidth}
\begin{document}
\title{\normalsize\Large\bfseries Homogeneous and heterogeneous 
nucleation in the three--state Blume--Capel model}

\author[1]{Emilio N.M.~Cirillo\footnote{emilio.cirillo@uniroma1.it}}
\author[2]{Vanessa~Jacquier\footnote{v.jacquier@uu.nl}}
\author[2]{Cristian~Spitoni\footnote{C.Spitoni@uu.nl}}

\affil[1]{Dipartimento di Scienze di Base e Applicate per l'Ingegneria, 
             Sapienza Universit\`a di Roma, 
             via A.\ Scarpa 16, I--00161, Roma, Italy.}

\affil[2]{Institute of Mathematics,
University of Utrecht, Budapestlaan 6, 3584 CD Utrecht.} 

\date{\empty} 

\maketitle

\begin{abstract}
The metastable behavior of the 
stochastic Blume--Capel model with Glauber dynamics is 
studied when zero-boundary conditions are considered. The presence of zero-boundary conditions changes drastically the metastability scenarios of the model: 
\emph{heterogeneous nucleation} will  be proven in the region of the parameter space where the chemical potential is  larger than the external magnetic field.
\end{abstract}

\keywords{Glauber dynamics, Blume--Capel model, metastability, nucleation,
          low temperature dynamics, effect of boundary conditions}

\begin{multicols}{2}

\section{Introduction}
\label{s:int}
\par\noindent
We study the metastable behavior of the stochastic Blume--Capel model 
\cite{blume1966,capel1966} under the Glauber dynamics with zero-boundary 
conditions.

The metastable behavior of the Blume--Capel model 
has been firstly rigorously studied in \cite{co1996} in finite volume in the 
limit of temperature tending to zero. In that paper the parameters have 
been chosen so that the metastable state is unique. 
The same regime is studied in 
\cite{cn2013,ll2016,cns2017} choosing the parameters in such a way that 
the model exhibits two not degenerate in energy metastable states
\cite{bet2021potts}.
The regime of infinite volume has been considered in 
\cite{mo2001,llm2019}.

Metastability is a widely studied phenomenon that has been investigated
on mathematical grounds 
in the past fifty years from different point of views and with 
several approaches. We will use, here, the so--called 
\emph{pathwise approach}, originally proposed in 
\cite{cassandro1984metastable} and developed in 
several more recent studies
\cite{olivieri1995markov,manzo2004essential,cirillo2015metastability,olivieri2005large}. 
This method provides a standard way to characterize metastable 
states and a technique to compute its exit time and to describe 
the typical exit trajectories. 

Two more approaches to the rigorous mathematical description of 
metastability have been developed in the last decades, 
the \emph{potential--theoretic approach}
\cite{mathieupicco,bovier2002metastability,bovier2016metastability} 
and the  
\emph{trace method} \cite{beltran2010tunneling}.

The study of metastability is typically conducted for periodic 
boundary conditions; these are, indeed, a rather natural setup in this 
context. In particular periodic boundary conditions were considered 
in all the studies of the Blume--Capel model mentioned above.
In the present paper we shall consider the case of zero-boundary 
condition, which is particularly important from the point 
of view of applications. Indeed, not periodic boundary conditions 
mimic the presence of a defect in the system.

In the presence of defects (or boundaries), the nucleus of the new phase forms in contact with the impurities (or boundaries), so that the properties of the impurities control the nucleation rate. The nucleation starts indeed at phase boundaries or impurities, since at these sites the
 free energy barrier is lower, and the nucleation is facilitated. Therefore, the nucleation observed in practice is usually catalyzed and it is named \emph{heterogeneous nucleation}: see for instance \cite{cantor2003, perepezko2003}  for the crystallization case, and \cite{smorodin2006} for the condensation.

 In order to understand the general mechanisms triggering the heterogeneous nucleation,  Monte Carlo simulations for simple toy models (e.g., Ising and lattice gas models) are often used, see \cite{binder2016}. For instance, in \cite{page2006} Monte Carlo simulations for two-dimensional Ising models are used for studying the roles of \emph{pores} on the surface in the nucleation process.  The simulations show that the nucleation occurring at pores has nucleation rate of orders of magnitude higher than the one starting on flat surfaces. This behavior is very common for porous materials, which often present indeed the well known phenomenon of \emph{capillary condensation}, i.e., the condensation of liquid bridges in the pores, see \cite{restagno2000}.

Simulations of two dimensional Ising model have also been used for studying the the role of microscopic impurities (i.e., sites with fixed spins) in the 
bulk \cite{sear2006}: the heterogeneous nucleation, starting from a single fixed spin, is more than four orders of magnitude faster than homogeneous
nucleation. Therefore, small microscopic impurities strongly promote nucleation, making very difficult to purify a sample sufficiently in order to observe homogeneous nucleation. The same conclusions are obtained as well for the two-dimensional Potts model \cite{sear2005} with competitive nucleating phases.

Heterogeneous nucleation plays also a pivotal role in the process of \emph{crystallization of proteins} on surfaces (see \cite{saridakis2009}): by tuning the
geometrical properties of the surface (porosity, pore size, roughness), heterogeneous nucleation can be activated, enhancing 
the probability of obtaining crystals with appropriate size. The paper \cite{curcio2010} uses a two-dimensional Ising model for showing the dependence of the nucleation rate on the the polymeric surfaces used as substrate for heterogeneous nucleation. Different rough surfaces are modeled indeed with different profiles of fixed spins at the boundaries.

The effect of the boundary conditions on the metastable behavior 
was studied on rigorous terms in \cite{cl1998}
in the framework of the Ising model; there the free boundary condition
case was considered.
The authors proved that the main features characterizing the 
metastable behavior in the case of periodic boundary conditions 
remain unchanged. But some new effects show up: the main difference 
with the periodic case is that the nucleation phenomenon is no more 
spatially homogeneous, in the sense that the critical 
droplet, which has to be formed to nucleate the stable state, 
appears necessarily at one of the four corners of the lattice. 
Other details are different, such as the size of the critical droplet and, 
consequently, the exponential estimate of the exit time. 

In this paper we shall show that, due to the three--state 
character of the Blume--Capel model, the metastability 
scenario proven for periodic boundary conditions \cite{co1996} changes 
deeply when different boundary conditions are considered. 
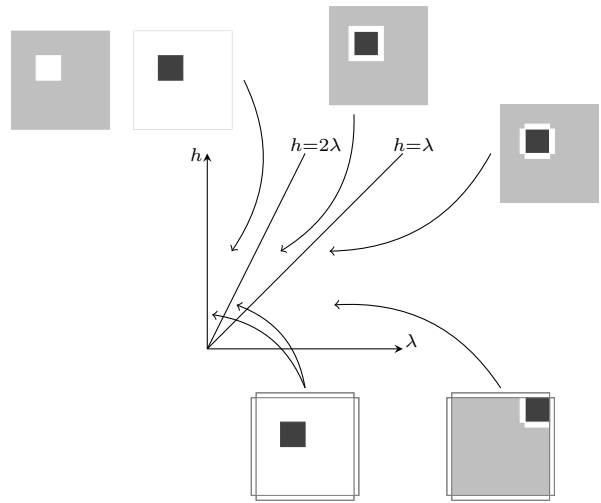
\begin{figure}[H]
\begin{center}
\begin{tikzpicture}[scale=0.65]
  \draw[-stealth](4,4)--(8,4);
  \draw[-stealth](4,4)--(4,8);
  \draw (4,4)--(8,8);
  \draw (4,4)--(6,8);
  \node[] at (6.,8.2) {${\phantom{m}_{h=2\lambda}}$};
  \node[] at (8.,8.2) {${\phantom{m}_{h=\lambda}}$};
  \node[] at (3.55,8.) {${\phantom{m}_{h}}$};
  \node[] at (7.95,4.2) {${\phantom{m}_{\lambda}}$};
  [very thick]
  \draw[->] (4.75,9.5) to [bend left] (4.5,6.);
  \draw[->] (7.,8.8) to [bend left] (5.5,6.);
  \draw[->] (9.8,8.) to [bend left] (6.5,6.);
  \draw[->] (6.,3.2) to [bend right] (4.1,4.7);
  \draw[->] (6.,3.2) to [bend right] (4.6,4.9);
  \draw[->] (10.,3.2) to [bend right] (6.6,4.9);

  \draw[fill=lightgray,lightgray] (0,8.5) rectangle (2,10.5);
  \draw[fill=gray,gray] (0.5,9.5) rectangle (1.,10);
  \draw[fill=gray,gray] (2.5,8.5) rectangle (4.5,10.5);
  \draw[fill=black,black] (3.0,9.5) rectangle (3.5,10);

  \draw[fill=lightgray,lightgray] (6.5,9.) rectangle (8.5,11.);
  \draw[fill=black,black] (7.0,10.) rectangle (7.5,10.5);
  \draw[fill=gray,gray] (6.9,10.5) rectangle (7.6,10.6);
  \draw[fill=gray,gray] (6.9,9.9) rectangle (7.,10.5);
  \draw[fill=gray,gray] (7,9.9) rectangle (7.6,10.);
  \draw[fill=gray,gray] (7.5,10) rectangle (7.6,10.5);

  \draw[fill=lightgray,lightgray] (10.,7.) rectangle (12.,9.);
  \draw[fill=black,black] (10.5,8.) rectangle (11.,8.5);
  \draw[fill=gray,gray] (10.5,8.5) rectangle (11.,8.6);
  \draw[fill=gray,gray] (10.5,7.9) rectangle (11.,8.);
  \draw[fill=gray,gray] (10.4,8.) rectangle (10.5,8.5);
  \draw[fill=gray,gray] (11.,8.) rectangle (11.1,8.5);

  \draw[fill=gray,gray] (5,1.) rectangle (7,3.);
  \draw[fill=black,black] (5.5,2.) rectangle (6.,2.5);
  \draw[gray] (5.,3.) rectangle (7.,3.1);
  \draw[gray] (5.,0.9) rectangle (7.,1.);
  \draw[gray] (4.9,1.) rectangle (5.,3.);
  \draw[gray] (7.,1.) rectangle (7.1,3.);

  \draw[fill=lightgray,lightgray] (9.,1.) rectangle (11.,3.);
  \draw[fill=black,black] (10.5,2.5) rectangle (11.,3.);
  \draw[fill=gray,gray] (10.5,2.4) rectangle (11.,2.5);
  \draw[fill=gray,gray] (10.4,2.5) rectangle (10.5,3.);
  \draw[gray] (9.,3.) rectangle (11.,3.1);
  \draw[gray] (9.,0.9) rectangle (11.,1.);
  \draw[gray] (8.9,1.) rectangle (9.,3.);
  \draw[gray] (11.,1.) rectangle (11.1,3.);


  \draw[fill=lightgray,lightgray] (0,8.5) rectangle (2,10.5);
  \draw[fill=white,white] (0.5,9.5) rectangle (1.,10);
  \draw[fill=white,white] (2.5,8.5) rectangle (4.5,10.5);
  \draw[fill=darkgray,darkgray] (3.0,9.5) rectangle (3.5,10);

  \draw[fill=lightgray,lightgray] (6.5,9.) rectangle (8.5,11.);
  \draw[fill=darkgray,darkgray] (7.0,10.) rectangle (7.5,10.5);
  \draw[fill=white,white] (6.9,10.5) rectangle (7.6,10.6);
  \draw[fill=white,white] (6.9,9.9) rectangle (7.,10.5);
  \draw[fill=white,white] (7,9.9) rectangle (7.6,10.);
  \draw[fill=white,white] (7.5,10) rectangle (7.6,10.5);

  \draw[fill=lightgray,lightgray] (10.,7.) rectangle (12.,9.);
  \draw[fill=darkgray,darkgray] (10.5,8.) rectangle (11.,8.5);
  \draw[fill=white,white] (10.5,8.5) rectangle (11.,8.6);
  \draw[fill=white,white] (10.5,7.9) rectangle (11.,8.);
  \draw[fill=white,white] (10.4,8.) rectangle (10.5,8.5);
  \draw[fill=white,white] (11.,8.) rectangle (11.1,8.5);

  \draw[fill=white,white] (5,1.) rectangle (7,3.);
  \draw[fill=darkgray,darkgray] (5.5,2.) rectangle (6.,2.5);
  \draw[gray] (5.,3.) rectangle (7.,3.1);
  \draw[gray] (5.,0.9) rectangle (7.,1.);
  \draw[gray] (4.9,1.) rectangle (5.,3.);
  \draw[gray] (7.,1.) rectangle (7.1,3.);

  \draw[fill=lightgray,lightgray] (9.,1.) rectangle (11.,3.);
  \draw[fill=darkgray,darkgray] (10.5,2.5) rectangle (11.,3.);
  \draw[fill=white,white] (10.5,2.4) rectangle (11.,2.5);
  \draw[fill=white,white] (10.4,2.5) rectangle (10.5,3.);
  \draw[gray] (9.,3.) rectangle (11.,3.1);
  \draw[gray] (9.,0.9) rectangle (11.,1.);
  \draw[gray] (8.9,1.) rectangle (9.,3.);
  \draw[gray] (11.,1.) rectangle (11.1,3.);
\end{tikzpicture}
\end{center}
\caption{Schematic representation of the behavior of the Blume--Capel model
in the region $h,\lambda>0$ in the case of periodic boundary
condition (top pictures) and in the case of zero-boundary
condition (bottom pictures). Light gray 
for minuses, dark gray for pluses, and white for zeros.}
\label{fig:int000}
\end{figure}
The Hamiltonian of the Blume--Capel model depends on two 
parameters, the magnetic field $h$ and the chemical 
potential $\lambda$. The spin variables can take three values, 
$-1$, $0$, and $+1$.
We limit our discussion to the 
case $\lambda,h>0$, where the chemical potential term 
equally favors minus and plus spins with respect to zeroes and 
the magnetic field favors pluses and disadvantages 
minuses with respect to the zeroes.
In this parameter region, in the periodic case, it was proven in 
\cite{co1996} the following result
(see Figure~\ref{fig:int000} for a schematic description): 
the stable state is the homogeneous plus state and the metastable state is
the homogeneous minus states. 
Moreover, for $h>2\lambda$ the system exits the metastable state 
via the formation of a zeroes square droplet and reaches the 
homogeneous zero state. Then, at a random time the transition 
from the zero state to the stable state is realized via the formation 
of a plus square droplet. 
For $2\lambda>h$ the system exits the metastable state 
via the formation of a plus square droplet separated by the sea of minuses 
by a layer of zeroes of width one (with minus at the corners 
in the case $\lambda>h$). In this way the stable state is directly 
reached. 

This scenario changes drastically when zero-boundary conditions 
are considered: for $h>\lambda$ the metastable state is the 
homogeneous zero state and the plus stable state is reached via 
the formation of a plus square droplet at any point in the lattice. 
For $\lambda>h$, on the contrary, the situation is similar to 
the periodic boundary condition case, but, starting from the 
minus metastabe state, the stable phase is nucleated 
at one of the four corners of the lattice
via the formation of a plus square droplet
separated by the sea of minus by a one site zero layer. 
Thus, the nucleation is spatially 
homogeneous for $h>\lambda$ and spatially not homogeneous for 
$\lambda>h$. 
This scenario will be proved rigorously in the region $\lambda>h$ of the 
parameter plane.

The paper is organized as follows. 
In Section~\ref{s:mrs} we introduce the model.
In Section~\ref{section:main_results} we state the main results. 
In particular in Section~\ref{s:heu} we present the heuristic 
study of the metastable behavior in the whole parameter 
region $h,\lambda>0$, while in Section~\ref{s:res} we state formal 
results for the restricted region $\lambda>h>0$. 
Section~\ref{s:pro-th} is devoted to the proof of the results
stated in Section~\ref{s:mrs} and \ref{section:main_results}, 
while the proofs of the more technical lemmas are reported in 
Section~\ref{s:pro-le}.

\section{Model and definitions}
\label{s:mrs}
\par\noindent
In this section we first define the model and then state our main 
results. Proofs are postponed to the following sections.

\subsection{The lattice}
\label{s:lat}
\par\noindent
We consider the set $\mathbb{Z}^2$ embedded in $\mathbb{R}^2$ 
and call \emph{sites} its elements.
Given two sites $i,i'\in\mathbb{Z}^2$ we let $|i-i'|$ be their 
Euclidian distance. 
Given $i\in\mathbb{Z}^2$, we say that 
$i'\in\mathbb{Z}^2$ is a \emph{nearest neighbor} 
of $i$ if and only if $|i-i'|=1$. 
Pairs of neighboring sites will be called \emph{bonds}.
A set $I\subset\mathbb{Z}^2$ is \emph{connected} if and 
only if for any $i\neq i'\in I$ there exists a sequence 
$i_1,i_2,\dots,i_n$ of sites of $I$ such that 
$i_1=i$, 
$i_n=i'$, and 
$i_k$ and $i_{k+1}$ are nearest neighbors for any $k=1,\dots,n-1$.

A \emph{column}, resp. a \emph{row} of $\mathbb{Z}^2$ as a sequence of vertical, resp. horizontal, connected sites.

Given $I\subset \mathbb{Z}^2$ we call 
\emph{internal boundary} $\partial^-I$ of $I$
the set of sites in $I$ having a nearest neighbor outside 
$I$. 
The \emph{bulk} of $I$ is the set $I\setminus\partial^-I$, 
namely, the set of sites of $I$ having four nearest neighbors 
in $I$. 
We call 
\emph{external boundary} $\partial^+I$ of $I$
the set of sites in $\mathbb{Z}^2\setminus I$ having 
a nearest neighbor inside $I$. 

A set $R\subset\mathbb{Z}^2$ is called a \emph{rectangle} 
(resp.~\emph{square}) if the union of the 
closed unit squares of $\mathbb{R}^2$ 
centered at the site of $R$ with sides parallel 
to the axes of $\mathbb{Z}^2$ is a rectangle (resp.~a square) 
of $\mathbb{R}^2$.
The \emph{sides} of a rectangle are the four maximal connected 
subsets of the its internal boundary lying on straight lines 
parallel to the axes of $\mathbb{Z}^2$.
The \emph{length} of one side of a rectangle is the number of sites 
belonging to the side itself. 
A \emph{quasi-square} is a rectangle with side lengths equal to $n$ and $ n+1$.

For any set $I\subset\mathbb{Z}^2$ we call \emph{rectangular envelope}
of $I$ the smallest (with respect to inclusion) rectangle 
$R\subset\mathbb{Z}^2$ such that $I\subset R$. 
Two rectangles of $\mathbb{Z}^2$ are called \emph{interacting} if there 
exists a site not belonging to them at distance one from both of them.
Given a finite set $I\subset\mathbb{Z}^2$, the \emph{bootstrap 
construction} associates to $I$ a collection of not interacting 
rectangles through the following sequence of 
operations:
i) partition $I$ in maximal connected subsets;
ii) Consider the family of rectangles obtained by collecting 
the rectangular envelope of each maximal connected subset of $I$;
iii) Partition the family of rectangles in maximal sequences 
of pairwise interacting rectangles;
iv) Consider a new family of rectangles obtained by collecting 
the rectangular envelope of the union of the rectangles of each 
of the maximal sequences constructed at point iii);
v) Repeat the operations iii) and iv) until the family of 
rectangles constructed at point iv) is made of pairwise not 
interacting rectangles. 

\subsection{The Blume--Capel model}
\label{s:mod}
\par\noindent
Consider the square $\Lambda=\{1,\dots,L\}^2\subset \mathbb{Z}^2$.
Let $\{-1,0,+1\}$
be the \emph{single spin state space} and $\mathcal{X}:=\{-1,0,+1\}^{\Lambda}$
be the \emph{configuration} or \emph{state} space. 
With $\puno$, $\muno$, $\zero$
we denote the homogeneous configurations
in which all the spins are equal to 
$+1$, $-1$, and $0$, respectively.
Let $\eta \in \mathcal{X}$ and $A \subseteq \Lambda$, we denote by $\eta_A$ the \emph{restricted configuration} of $\eta$ on the subset $A$. 
We say that two configurations $\sigma$ and $\eta$ are 
\emph{communicating} if and only if they differ at most for the value 
of the spin at one site, and we denote by $\sigma \sim \eta$. 

The \emph{Hamiltonian} of
the model is
\begin{equation}
\label{mod005}
\begin{array}{rcl}
H(\eta)
&\!\!=&\!\!
{\displaystyle
\frac{J}{2}\sum_{\newatop{i,j\in\Lambda:}{|i-j|=1}}[\eta(i)-\eta(j)]^2
+
J\sum_{i\in\partial^-\Lambda}
 \sum_{\newatop{j\in\mathbb{Z}^2\setminus\Lambda:}{|i-j|=1}}
[\eta(i)]^2
}
\\
&\!\!&\!\!
{\displaystyle
-\lambda\sum_{i\in\Lambda}\eta(i)^2
-h\sum_{i\in\Lambda}\eta(i)
}
\end{array}
\end{equation}
for any $\eta\in\mathcal{X}$, where
$J>0$ is called the \emph{coupling constant}, 
$\lambda,h\in\mathbb{R}$ are  called \emph{chemical potential}
and 
\emph{magnetic field} respectively. 
The first term at the right--hand side of \eqref{mod005} 
will be called \emph{internal interaction} term, 
the second term  \emph{boundary interaction} term, and 
the last two will 
be called \emph{site} terms.
We stress that the second term accounts for 
the interaction between the spins at the sites of the internal boundary 
of $\Lambda$ and the zero external boundary conditions: each site of the
internal boundary contributes with one single bond, excepted for the 
four sites at the corners of $\Lambda$, which contributes with 
two bonds each. 
We will refer to $H(\eta)$ as the \emph{energy} of the configuration $\eta$.

In order to state our results we will rely on the following assumptions 
on the parameters of the model\footnote{With the notation $0<a\ll b$ we mean
$0<a< cb$ for some suitable positive constant $c$ that we are not 
interested to compute exactly.}.

\begin{condition}
\label{c:mod007}
We assume that the parameters of the model satisfy the following 
properties:

\begin{enumerate}
\item $J\gg\lambda,h>0$, \label{mod007a}
\item $L>\Big(\frac{2J}{\lambda-h}\Big)^3$, \label{mod007b}
\item $\frac{2J}{\lambda+h}, \, \frac{2J}{\lambda-h}, \, \frac{2J+\lambda-h}{\lambda+h}, \, \frac{J+\lambda+h}{h}$ are not integers. \label{mod007c}
\end{enumerate}
%
\end{condition}


The Gibbs measure
associated with the Hamiltonian \eqref{mod005} is
\begin{equation}
\label{mod010}
\mu_\beta(\eta)=\frac{1}{Z_\beta}\exp\{-\beta H(\eta)\}
\end{equation}
where $Z_\beta=\sum_{\eta'\in\cc{X}} \exp\{-\beta H(\eta')\}$
is the \emph{partition function} 
and $\beta>0$ is the inverse \emph{temperature}.

The time evolution of the model will be 
defined by assuming that spins evolve according to a Glauber dynamics, 
with the Metropolis weights.
More precisely, 
we consider the discrete time Markov chain $\sigma_t\in\mathcal{X}$, 
with $t\ge0$, with transition matrix $p_\beta$ defined as follows: 
$p_\beta(\eta,\eta')=0$ for 
$\eta,\eta'\in\mathcal{X}$ not communicating configurations,
\begin{equation}
\label{mod015}
p_\beta(\eta,\eta')
=
 \frac{1}{2|\Lambda|}e^{-\beta[H(\eta')-H(\eta)]_+}
\end{equation}
for $\eta,\eta'\in\mathcal{X}$ communicating configuration such 
that $\eta\neq\eta'$
(where, for any real $a$, we let $[a]_+=a$  if $a>0$ and $0$ if $a<0$),
and 
\begin{equation}
\label{mod020}
p_\beta(\eta,\eta)
=1-\sum_{\eta'\neq\eta}p_\beta(\eta,\eta')
\end{equation}
for any $\eta\in\mathcal{X}$.
The dynamics can be described as follows: at each time a site is chosen 
with uniform probability $1/|\Lambda|$ and a spin value differing from the one 
at the chosen site is selected with probability $1/2$, then 
the flip of the spin at the chosen site 
to the selected spin value is performed with the Metropolis 
probability. 

The probability measure induced by the Markov chain 
started at $\eta$ is denoted by 
$P_\eta$ and the related 
expectation is 
denoted by $E_\eta$.
 
\begin{lemma}
\label{t:mod000}
The Markov chain defined above
is reversible with respect to the Gibbs measure 
\eqref{mod010}, i.e.,
the detailed balance condition 
\begin{equation}
\label{mod050}
\mu_\beta(\eta)p_\beta(\eta,\eta')
=
\mu_\beta(\eta')p_\beta(\eta',\eta)
\end{equation}
is satisfied
for any $\eta,\eta'\in\mathcal{X}$.
\end{lemma}

\subsection{Paths, energy costs, metastable states}
\label{s:pat}
\par\noindent
A sequence of configurations 
$(\omega_1,\omega_2,\dots,\omega_n)\in\mathcal{X}^n$ such that 
$\omega_i$ and $\omega_{i+1}$ are communicating for any $i=1,2,\dots,n-1$ 
is called a \emph{path of length} $n$.
A path $(\omega_1,\dots,\omega_n)$ is called \emph{downhill}
(resp.~\emph{uphill}) if and only if 
$H(\omega_i)\ge H(\omega_{i+1})$ 
(resp.~$H(\omega_i)\le H(\omega_{i+1})$) 
for any $i=1,2,\dots,n-1$.
In particular, a path $(\omega_1,\dots,\omega_n)$ is called \emph{two-steps downhill}
 if and only if $H(\omega_i)\ge H(\omega_{i+2}) \ge H(\omega_{i+1})$ 
for any $i=1,2,\dots,n-2$.
Given two configurations $\eta,\eta'\in\mathcal{X}$, the set of 
paths with first configuration $\eta$ and last configurations $\eta'$ 
is denoted by $\Omega(\eta,\eta')$.

Given a path $\underline{\omega}=(\omega_1,\dots,\omega_n)$, its \emph{height} 
$\Phi(\underline{\omega})$ is the maximal height reached by the 
configurations of the path, more precisely,
\begin{equation}
\label{pat000}
\Phi(\underline{\omega})
=\max_{i=1,\dots,n}H(\omega_i)
\;.
\end{equation}
Given two configurations $\eta,\eta'$, the \emph{communication height}
between $\eta$ and $\eta'$ is defined as 
\begin{equation}
\label{pat005}
\Phi(\eta,\eta')
=
\min_{\omega\in\Omega(\eta,\eta')}\Phi(\omega)
\;.
\end{equation}
Any path $\omega\in\Omega(\eta,\eta')$ such that 
$\Phi(\omega)=\Phi(\eta,\eta')$ is called \emph{optimal} 
for $\eta$ and $\eta'$.

The \emph{stability level} of
a configuration $\eta \in \mathcal{X}$ is
\begin{equation}
V_{\eta}:=\Phi(\sigma,\mathcal{I}_{\eta})-H(\eta),
\end{equation}
where $\mathcal{I}_{\sigma}$ is the set of configurations with energy strictly lower than $H(\eta)$. 
If $\mathcal{I}_{\sigma}$ is empty, then we define $V_{\sigma}=\infty$. 

The metastable states are those states where the stability level is maximum. We denote by $\Gamma_m$ the \emph{maximal stability level},
\begin{equation}\label{Gamma}
    \Gamma_m:=\max_{\sigma\in \mathcal{X}\setminus \mathcal{X}^s}V_{\sigma}.
\end{equation}

Moreover, we define the \emph{energy barrier}
as $\Phi(m, s)-H(m)$, where $m$ is a metastable state and $s$ is a ground state.

\subsection{Energy landscape}
\label{s:lan}
\par\noindent
A crucial ingredient for several results 
discussed in this section is the value of
the energy difference (\emph{energy cost}) associated with each possible spin flip.

\begin{table*}
\begin{center}
\begin{tabular}{c|c|c|c|c|c}
\hline\hline
minuses & zeroes & pluses & minus to zero & minus to plus & zero to plus\\
\hline\hline
4 & 0 & 0 & $4J+\lambda-h$ & $16J-2h$ & $12J-\lambda-h$ \\
\hline
3 & 1 & 0 & $2J+\lambda-h$ & $12J-2h$ & $10J-\lambda-h$ \\
\hline
3 & 0 & 1 &   $+\lambda-h$ & $8J-2h$ & $8J-\lambda-h$ \\
\hline
2 & 2 & 0 &   $+\lambda-h$ & $8J-2h$ & $8J-\lambda-h$ \\
\hline
2 & 1 & 1 & $-2J+\lambda-h$ & $4J-2h$ & $6J-\lambda-h$ \\
\hline
2 & 0 & 2 & $-4J+\lambda-h$ & $-2h$ & $4J-\lambda-h$ \\
\hline
1 & 3 & 0 & $-2J+\lambda-h$ & $4J-2h$ & $6J-\lambda-h$ \\
\hline
1 & 2 & 1 & $-4J+\lambda-h$ & $-2h$ & $4J-\lambda-h$ \\
\hline
1 & 1 & 2 & $-6J+\lambda-h$ & $-4J-2h$ & $2J-\lambda-h$ \\
\hline
1 & 0 & 3 & $-8J+\lambda-h$ & $-8J-2h$ & $-\lambda-h$ \\
\hline
0 & 4 & 0 & $-4J+\lambda-h$ & $-2h$ & $4J-\lambda-h$ \\
\hline
0 & 3 & 1 & $-6J+\lambda-h$ & $-4J-2h$ & $2J-\lambda-h$ \\
\hline
0 & 2 & 2 & $-8J+\lambda-h$ & $-8J-2h$ & $-\lambda-h$ \\
\hline
0 & 1 & 3 & $-10J+\lambda-h$ & $-12J-2h$ & $-2J-\lambda-h$ \\
\hline
0 & 0 & 4 & $-12J+\lambda-h$ & $-16J-2h$ & $-4J-\lambda-h$ \\
\hline\hline
\end{tabular}
\end{center}
\caption{Energy difference for a spin flip for all neighbor
configurations (opposite sign for reversed flip). 
The number of neighbor minuses, zeroes, and pluses
is reported in the first three columns and the energy difference 
in the last three.
For flips at the boundary (resp. corners) see the rows with at least 
one (resp. two) zero among the nearest neighbors.
}
\label{tab:heu000}
\end{table*}
The energy differences for a spin flip for all neighbor configurations  obtained from \eqref{mod005} are listed 
in table~\ref{tab:heu000}.
Since, as noted above, the boundary interaction term
is equal to the internal interaction with fixed zero condition 
in the external boundary, 
the energy difference associated with possible spin flips at the boundary 
is given by the rows of table~\ref{tab:heu000} with at least 
one zero among the nearest neighbors (at least two for the flip 
of a spin at the corners of $\Lambda$).

As we will see below, the homogeneous states $\muno$, $\zero$, and $\puno$ 
will play a crucial role in our study. 
We remark that, 
from \eqref{mod005}, it follows
\begin{equation}
\label{lan000}
H(\pmuno)=4JL-|\Lambda|(\lambda\pm h)
\,\textup{ and }\, H(\zero)=0.
\end{equation}
Thus, under the assumptions \eqref{mod007a} and \eqref{mod007b}, the 
energy hierarchy of the homogeneous states is
\begin{equation}
\label{lan002}
H(\puno)<H(\zero)<H(\muno)
\;\;\;\;\textup{ for } h\ge\lambda
\end{equation}
and 
\begin{equation}
\label{lan004}
H(\puno)<H(\muno)<H(\zero)
\;\;\;\;\textup{ for } h<\lambda.
\end{equation}

The \emph{ground state} of the system (or of the Hamiltonian)
is the configuration where 
the Hamiltonian \eqref{mod005} attains its absolute 
minimum\footnote{We note that if the second term at the right--hand side
of \eqref{mod005} was not present, then, both for free and periodic 
boundary conditions, for $\lambda>0$,
the ground state would be the plus homogeneous configuration $\puno$
for $h>0$ and 
the minus homogeneous configuration $\muno$
for $h<0$. This would follow 
from the fact that in these homogeneous states the 
interaction contribution to the Hamiltonian is zero and the site 
contribution is minimal.}.
\begin{lemma}
\label{t:lan00-5}
Under Condition~\ref{c:mod007} 
the homogeneous state $\puno$ 
is the ground state of the system.
\end{lemma}

We say that a configuration $\eta\in\mathcal{X}$ 
is a \emph{local minimum} of the Hamiltonian if and only if 
for any $\eta'\neq\eta$ communicating with $\eta$ we have 
$H(\eta')>H(\eta)$.
Important examples of local minima, in suitable regions of the 
parameter plane $\lambda$--$h$,  are the homogeneous states. 
We make this remark rigorous in the following lemma. 

\begin{lemma}
\label{t:lan000}
Assume \eqref{mod007a} is satisfied.
For $h>\lambda$ 
the homogeneous state $\zero$ 
is a local minimum of the Hamiltonian. 
For $h<\lambda$ 
the homogeneous states $\zero$ and $\muno$ 
are local minima of the system. 
\end{lemma}

We stress that for $h>\lambda$ the state $\muno$ is not a local minimum, 
indeed, from row 1 in table~\ref{tab:heu000}, it follows that 
the four corner spins can be flipped to zero by decreasing the energy. 
Moreover, by repeating similar flips a downhill path from $\muno$ to 
$\zero$ can be constructed. 

Based on the above lemma, 
at the heuristic level, we can expect that 
the homogeneous states $\muno$ and $\zero$
are potential metastable states in the region of the parameter plane 
considered in the lemma. 

\section{Main results}
\label{section:main_results}
\par\noindent
In this section, we present the main results of the model. However, in Section \ref{s:heu} we first use some preliminary heuristic arguments  for describing the general metastable behavior in the region $0<h,\lambda\ll J$. Afterwards, we will state the actual theorem in the subregion $h>\lambda>0$.
\,\,
\newpage
\subsection{Heuristic discussion}\label{s:heu}
\par\noindent
We approach the heuristic study of the 
Blume--Capel model with zero-boundary conditions 
in the whole region $0<h,\lambda\ll J$.
We will have to distinguish several subregions where the 
metastable behavior will show peculiar features. 

This analysis is based on a very simple idea: the homogeneous states,
if local minima of the Hamiltonian, are potential metastable states 
of the system. When several possible metastable states are present, 
the true one is the one from which the system has to overcome 
the largest barrier to reach the stable state.
In order to compute such a barrier we imagine that the transition 
is realized through a sequence of local minima in which a droplet 
of stable phase grows in the sea of the metastable one. 

\subsubsection{Region $h>\lambda>0$}
\label{s:hla}
\par\noindent
In view of Lemma~\ref{t:lan000}
we are interested in the structures that give rise to local minima 
with zero background.

From rows 13--15 of table~\ref{tab:heu000} it follows that 
a configuration in which the sites with plus spin 
form a rectangle plunged in the sea of zeroes is a local minimum. 
We stress that the rectangular plus droplet can be located 
at one corner of the lattice $\Lambda$. 
We add that if the shape of the plus region is not a rectangle, 
then, since there exists at least a zero with more than two neighboring 
pluses, from rows 13--15 of tables~\ref{tab:heu000} it follows that 
the configuration is not a local minimum.

The energy of a square plus droplet of side length $\ell$ plunged 
in the sea of zeroes with respect to the energy of $\zero$ 
is $4JL-(\lambda+h)\ell^2$. 
Since its maximum is attained at $2J/(\lambda+h)$, we can infer 
that this is the critical length, in the sense 
that droplets with side length smaller than $2J/(\lambda+h)$
tend to shrink, otherwise they tend to grow. 
Moreover, we note that the difference of energy between 
the critical droplet and the configuration $\zero$ 
is $4J^2/(\lambda+h)$.

At the level of our very rough heuristic discussion, we can conclude that
the metastable state is the $\zero$ configuration, 
the transition to the stable state is performed via the 
nucleation of a square droplet of pluses 
of side length $2J/(\lambda+h)$ at any site of the lattice $\Lambda$ 
(homogeneous nucleation), 
and the exit time 
is of order $\exp\{\beta 4J^2/(\lambda+h)\}$.

\subsubsection{Region $\lambda>h>0$}
\label{s:hlo}
\par\noindent
In view of Lemma~\ref{t:lan000}
we are interested in the structures that give rise to local minima 
with zero or minus background. 

In the case of zero background, the same discussion as 
in Section~\ref{s:hla} suggests that the system can 
exit the state $\zero$ by overcoming the energy 
barrier $4J^2/(\lambda+h)$ and reaching the stable state $\puno$ 
via the formation of a critical square droplet of pluses 
with side length $2J/(\lambda+h)$. 
But also the possibility that the system abandons $\zero$ reaching 
$\muno$ must be explored: 
from rows 1, 2, and 4 of table~\ref{tab:heu000} it follows that 
a configuration in which the sites with minus spin 
form a rectangle plunged in the sea of zeroes is a local minimum. 
The energy of a square minus droplet of side length $\ell$ plunged 
in the sea of zeroes with respect to the energy of $\zero$ 
is $4JL-(\lambda-h)\ell^2$. 
The critical length is 
$2J/(\lambda-h)$ and the difference of energy between 
the critical droplet and the configuration $\zero$ 
is $4J^2/(\lambda-h)$.
Since in this parameter region $4J^2/(\lambda+h)<4J^2/(\lambda-h)$ we can 
conclude that the system, starting from $\zero$, will perform 
a direct transition to the stable state $\puno$ paying the 
energy cost $4J^2/(\lambda+h)$. 

For what concerns the minus background case, we note\footnote{We also 
note that a rectangle of zeroes in the sea of minuses 
is not a local minimum, since (see row 4 of table~\ref{tab:heu000})
the flip to minus of one zero with two zeroes and two minuses among its
neighbors (corner) decreases the energy of the configuration.
But this remark is not relevant from the metastability point of view, 
since, in view of \eqref{lan004}, the transition from $\muno$ to $\zero$ 
is of no interest in this region of the parameters.}
that a rectangle of pluses in the sea of minuses 
is not a local minimum, since (see row 6 of table~\ref{tab:heu000})
the flip to zero of one plus with two pluses and two minuses among its
neighbors (corner) decreases the energy of the configuration.

Some relevant structures that are local minima 
are reported in Figure~\ref{fig:app000}. 
To prove that the depicted structures are local minima the reader 
can use table~\ref{tab:heu000}.
The five structures in the figure will be addressed in the sequel as 
(a) \emph{frame},
(b) \emph{boundary frame},
(c) \emph{corner frame},
(d) \emph{chopped corner frame},
(e) \emph{chopped boundary frame}.

For each structure we compute its energy with respect to $\muno$ as 
a function of the side length $\ell$ of the internal plus square.
With an intuitive notation we have:
\begin{align}\label{eq:energy_frame}
\Delta_\textup{a}(\ell)
=&
-2h\ell^2
+4J\ell
+4J(\ell+2)
+4\ell(\lambda-h), \notag
\\
\Delta_\textup{b}(\ell)
=&
-2h\ell^2
+4J\ell
+2J(\ell+2)
+(4\ell+2)(\lambda-h), \notag
\\
\Delta_\textup{c}(\ell)
=&
-2h\ell^2
+4J\ell
+(4\ell+3)(\lambda-h), \notag
\\
\Delta_\textup{d}(\ell)
=&
-2h\ell^2
+2J\ell
+2J(\ell+1)-2J
+2\ell(\lambda-h), \notag
\\
\Delta_\textup{e}(\ell)
=&
-2h\ell^2
+3J\ell
+J(3\ell+2)
+3\ell(\lambda-h).
\end{align}
Now, we note that 
\begin{align}\label{eq:difference_energy_frame}
\Delta_\textup{a}-\Delta_\textup{d}
=&
4J\ell
+8J
+2\ell(\lambda-h), \notag
\\
\Delta_\textup{b}-\Delta_\textup{d}
=&
2J\ell
+4J
+(2\ell+2)(\lambda-h), \notag
\\
\Delta_\textup{c}-\Delta_\textup{d}
=&
(2\ell+3)(\lambda-h), \notag
\\
\Delta_\textup{e}-\Delta_\textup{d}
=&
2J\ell
+2J
+\ell(\lambda-h).
\end{align}
Since these differences are all positive, we can conclude that 
the mechanism providing the transition from $\muno$ to $\puno$ 
is the formation and growth of a chopped corner droplet. 

\begin{figure}[H]
\begin{center}
\begin{tikzpicture}[scale=0.7]
  \draw[fill=lightgray,lightgray] (0,0) rectangle (8,8);
  \draw[color=gray] (0,8) rectangle (8,8.2);
  \draw[gray] (-0.2,0) rectangle (0,8);
  \draw[gray] (0,-0.2) rectangle (8,0);
  \draw[gray] (8,0) rectangle (8.2,8);

  \draw[fill=black,black] (2,1) rectangle (3,2);
  \draw[fill=gray,gray] (2,2) rectangle (3,2.2);
  \draw[fill=gray,gray] (1.8,1) rectangle (2,2);
  \draw[fill=gray,gray] (2,0.8) rectangle (3,1);
  \draw[fill=gray,gray] (3,1) rectangle (3.2,2);
  \node[] at (3.7,1.5) {(a)};
 
  \draw[fill=black] (0.2,4) rectangle (1.2,5);
  \draw[fill=gray,gray] (0.2,5) rectangle (1.2,5.2);
  \draw[fill=gray,gray] (0,3.8) rectangle (0.2,5.2);
  \draw[fill=gray,gray] (0.2,3.8) rectangle (1.2,4);
  \draw[fill=gray,gray] (1.2,4) rectangle (1.4,5);
  \node[] at (1.9,4.5) {(b)};
 
  \draw[fill=black] (0.2,6.8) rectangle (1.2,7.8);
  \draw[fill=gray,gray] (0.2,7.8) rectangle (1.4,8);
  \draw[fill=gray,gray] (0,6.6) rectangle (0.2,8);
  \draw[fill=gray,gray] (0.2,6.6) rectangle (1.2,6.8);
  \draw[fill=gray,gray] (1.2,6.8) rectangle (1.4,7.8);
  \node[] at (1.9,7.3) {(c)};

  \draw[fill=black,black] (7,7) rectangle (8,8);
  \draw[fill=gray,gray] (6.8,7) rectangle (7,8);
  \draw[fill=gray,gray] (7,6.8) rectangle (8,7);
  \node[] at (6.3,7.5) {(d)};

  \draw[fill=black,black] (7,3) rectangle (8,4);
  \draw[fill=gray,gray] (7,4) rectangle (8,4.2);
  \draw[fill=gray,gray] (6.8,3) rectangle (7,4);
  \draw[fill=gray,gray] (7,2.8) rectangle (8,3);
  \node[] at (6.3,3.5) {(e)};

  \draw[fill=lightgray,lightgray] (0,0) rectangle (8,8);
  \draw[color=gray] (0,8) rectangle (8,8.2);
  \draw[gray] (-0.2,0) rectangle (0,8);
  \draw[gray] (0,-0.2) rectangle (8,0);
  \draw[gray] (8,0) rectangle (8.2,8);

  \draw[fill=darkgray,darkgray] (2,1) rectangle (3,2);
  \draw[fill=white,white] (2,2) rectangle (3,2.2);
  \draw[fill=white,white] (1.8,1) rectangle (2,2);
  \draw[fill=white,white] (2,0.8) rectangle (3,1);
  \draw[fill=white,white] (3,1) rectangle (3.2,2);
  \node[] at (3.7,1.5) {(a)};
 
  \draw[fill=darkgray] (0.2,4) rectangle (1.2,5);
  \draw[fill=white,white] (0.2,5) rectangle (1.2,5.2);
  \draw[fill=white,white] (0,3.8) rectangle (0.2,5.2);
  \draw[fill=white,white] (0.2,3.8) rectangle (1.2,4);
  \draw[fill=white,white] (1.2,4) rectangle (1.4,5);
  \node[] at (1.9,4.5) {(b)};
 
  \draw[fill=darkgray] (0.2,6.8) rectangle (1.2,7.8);
  \draw[fill=white,white] (0.2,7.8) rectangle (1.4,8);
  \draw[fill=white,white] (0,6.6) rectangle (0.2,8);
  \draw[fill=white,white] (0.2,6.6) rectangle (1.2,6.8);
  \draw[fill=white,white] (1.2,6.8) rectangle (1.4,7.8);
  \node[] at (1.9,7.3) {(c)};

  \draw[fill=darkgray,darkgray] (7,7) rectangle (8,8);
  \draw[fill=white,white] (6.8,7) rectangle (7,8);
  \draw[fill=white,white] (7,6.8) rectangle (8,7);
  \node[] at (6.3,7.5) {(d)};

  \draw[fill=darkgray,darkgray] (7,3) rectangle (8,4);
  \draw[fill=white,white] (7,4) rectangle (8,4.2);
  \draw[fill=white,white] (6.8,3) rectangle (7,4);
  \draw[fill=white,white] (7,2.8) rectangle (8,3);
  \node[] at (6.3,3.5) {(e)};
 
\end{tikzpicture}
\end{center}
\caption{Representation of local minima in the sea of minuses. Light gray 
for minuses, dark gray for pluses, and white for zeros.}
\label{fig:app000}
\end{figure}
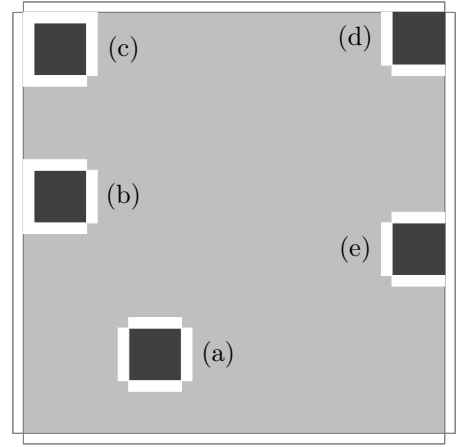

The length $\ell$ maximizing the energy (critical length) of such droplet 
is 
$[2J+(\lambda-h)]/(2h)$
and 
the energy of the critical droplet, with respect to $\muno$, 
in the limit $0<h<\lambda\sim0$ is
$2J^2/h$.

At the level of this heuristic analysis, 
it seems that 
the mechanism of the chopped corner frame is the best 
to perform the transition from the homogeneous $\muno$ state to the 
stable state $\puno$. 
This
transition is performed via the 
nucleation of a chopped corner frame 
of internal side length $[2J+(\lambda-h)]/(2h)$ 
(not homogeneous nucleation) 
and the exit time 
is of order $\exp\{\beta 2J^2/h\}$.
To establish which, between $\muno$ and $\zero$,
is the metastable state in the region $0<h<\lambda$
we note that in this region of the parameter plane 
$4J^2/(\lambda+h)<2J^2/h$ and, so,
the metastable state is $\muno$.
Moreover, we remark some relevant facts: 
the transition from the metastable to the stable state 
is direct, the nucleation is not homogeneous, 
and the exit time does not depend on $\lambda$.

\subsection{Main results for the region $\lambda>h>0$}\label{s:res}
\par\noindent
In the rest of the paper, we present the main results for the model in the region $0<h<\lambda$.
. 
The first theorem states that every configuration of 
$\mathcal{X}$ different from $\{\muno,\puno\}$ 
has a stability level strictly lower than $\Gamma$, 
where $\Gamma$ is the energy barrier to reach 
$\puno$ starting from $\muno$, i.e. $\Gamma=H(\sigma_s)-H(\muno)$
where $\sigma_s$ is the critical configuration represented 
in Figure \ref{fig:critical_configuration_saddle}. In particular,
\begin{equation}\label{eq:energy_barrier}
\Gamma=4Jl_c +2\lambda l_c-2h2l_c^2-2h
\end{equation}
where 
\begin{equation}\label{eq:critical_length}
    l_c=\lfloor\frac{2J+\lambda-h}{2h}\rfloor+1.
\end{equation}
\begin{figure}[H]
\begin{center}
    \includegraphics[scale=0.4]{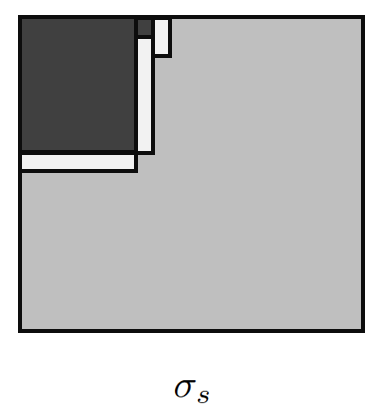}
    \caption{The critical configuration $\sigma_s$ which contains a critical chopped corner frame with side length $l_c$.}\label{fig:critical_configuration_saddle}
\end{center}
    \end{figure}
\begin{proposition}\label{thm:out_gammaBCG}
Let $\eta \in \mathcal{X}$ be a configuration 
such that $\eta \not \in \{\muno,\puno\}$, 
then $V_{\eta} <\Gamma$.
\end{proposition}
This result suggests that the only configurations with a stability level 
greater than or equal to $\Gamma$ could be
$\muno$, $\puno$. This is confirmed by Theorem \ref{thm:Identification}, 
where we identify the unique metastable state $\muno$ 
and the stable state $\puno$ in the region 
$\lambda>h>0$.

In the following theorem, we state
the recurrence of the system to the set $\{\muno, \puno\}$. In particular, Equation \eqref{eq:recurrence_BCG} implies that the system reaches with high probability either the state $\muno$ (which is a local minimizer
of the Hamiltonian) or the  ground state in a time shorter than 
$e^{\beta (\Gamma+\epsilon)}$, uniformly in the starting configuration 
$\eta$ for any $\epsilon >0$. In other words we can say that the dynamics speeded up by a 
time factor of order $e^{\beta \Gamma}$ reaches with high probability 
$\{ \muno, \puno \}$. 
\begin{theorem}[Recurrence property]\label{thm:recurrence_property_BCG} 
For any $\epsilon>0$ and sufficiently large $\beta$, the function
\begin{align}\label{eq:recurrence_BCG}
    \beta \to \sup_{\eta \in \mathcal{X}} \mathbb{P}_{\eta}(\tau_{\{\muno, \puno\}}> e^{\beta(\Gamma+\epsilon)})
    \end{align}
    is SES\footnote{We say that a function $\beta \mapsto f(\beta)$ is super exponentially small (SES) if $\lim_{\beta\to\infty}\frac{\log{f(\beta)}}{\beta}=-\infty.$}.
\end{theorem}

In the next theorem we identify the metastable state and we compute the maximal stability level. Recalling the
$\Gamma$ in \eqref{eq:energy_barrier}, we have

\begin{theorem}\label{thm:Identification} (Identification of metastable state)
In the region $\lambda>h>0$, the unique metastable state is $\muno$ and 
$\Gamma_m=\Gamma$.
\end{theorem}

Last goals is finding the asymptotic behavior as $\beta \to \infty$ 
of the transition time for the system started at the metastable state $\muno$. 
\begin{theorem}[Asymptotic behavior of $\tau_{\puno}$ in probability]\label{thm:transition_time_BCG}
For any $\epsilon>0$, we have
\begin{equation}
    \lim_{\beta \to \infty} \mathbb{P}_{\muno}(e^{\beta(\Gamma-\epsilon)}< \tau_{\puno}<e^{\beta(\Gamma+\epsilon)}) =1.
\end{equation}
\end{theorem}

\section{Proof of main results}\label{s:pro-th}
\par\noindent
In this section we collect the proofs of all the lemmas 
stated in Section~\ref{s:mrs} and of all theorems stated in Section \ref{section:main_results}.

\subsection{Proof of Lemma~\ref{t:mod000}}
\par\noindent
The statement is trivial
in the cases $\eta$ and $\eta'$ not communicating
and $\eta'=\eta$. Thus, suppose $\eta\neq\eta'$ are 
communicating: if $H(\eta)=H(\eta')$ then 
\eqref{mod050} is immediate, on the other hand 
if $H(\eta')>H(\eta)$ (the opposite case can be treated analogously) 
the statements follows from the definition of the Gibbs measure
\eqref{mod010} and the fact that 
\begin{displaymath}
p_\beta(\eta,\eta')
=
\frac{1}{2|\Lambda|}e^{-\beta[H(\eta')-H(\eta)]}
\textup{ and }
p_\beta(\eta',\eta)
=
\frac{1}{2|\Lambda|}
.
\end{displaymath}
\qed

\subsection{Proof of Lemma~\ref{t:lan00-5}}
\par\noindent
Recall we assumed that Condition~\ref{c:mod007} is in force. 

Case 1: pick a configuration $\eta\neq\puno$, 
such that there is at least 
a minus spin. 
Consider the configuration 
$\eta'$ obtained by flipping in $\eta$ all
the minuses to plus. 
We $H(\eta')<H(\eta)$, indeed, 
i) the internal interaction term 
at the right--hand side of \eqref{mod005}
is smaller for $\eta'$ since 
nothing changes for the bonds between minus spins of $\eta$ and 
for the bonds in which, in $\eta$,
one site has spin minus and the other has spin zero, 
on the other hand the interaction decreases if, in $\eta$, 
one of the sites of the bond has spin minus and the other has spin plus;
ii) the boundary interaction term is the same in $\eta$ and $\eta'$;
iii) the chemical potential term in $\eta'$ is the same as the one in $\eta$;
iv) the magnetic field in $\eta'$ is smaller than that in $\eta$ 
by the amount $2h$ for each flipped spin. 
If $\eta'=\puno$ the proof is over, otherwise there exists in $\eta'$ 
at least a zero spin and the proof is completed in the following case.

Case 2:
consider a configuration $\eta'\neq\puno$, 
such that there is no minus spin.
Consider 
the configuration $\eta''$ obtained by flipping to plus all 
the zero spins in $\eta'$ associated with the sites belonging to one of the 
not interacting rectangles obtained by applying the bootstrap 
construction (see Section~\ref{s:lat}) 
to the set of sites where $\eta'$ is plus one. 
If $\eta''\neq\eta'$ then $H(\eta'')<H(\eta')$ because it is possible 
to construct a downhill path from $\eta'$ to $\eta''$ such that 
at each step a zero spin with at least two neighboring plus sites
and no neighboring minus is flipped to plus decreasing 
the energy of the configurations 
(see rows 13--15 in table~\ref{tab:heu000}).
If $\eta''=\puno$ the proof is over.
In case $\eta''\neq\puno$, let $\ell$ be the largest side length 
of the rectangles in which $\eta''$ is plus one: 

Case 2.1: suppose $\ell<2J/(h+\lambda)$. 
Consider the configuration $\eta'''$ obtained by flipping to zero 
all the pluses in one of the sides of length $\ell$. 
From \eqref{mod005} we get
$H(\eta''')-H(\eta'')=-2J+(\lambda+h)\ell$, which implies 
$H(\eta''')<H(\eta'')$.
By removing one side after the other we prove $H(\zero)<H(\eta')$ and, 
from \eqref{lan002}, which is valid under the hypotheses of 
this lemma, we get 
$H(\puno)<H(\eta')$.

Case 2.2: suppose $\ell>2J/(h+\lambda)$. Now, consider one of the rectangles
on which $\eta''$ is plus one with maximal side length  
equal to $\ell$. Consider the configuration $\eta''''$ obtained 
by flipping to plus all the zeros associated with sites neighboring 
one of the sides of this rectangle whose length is equal 
to $\ell$. 
From \eqref{mod005} we get
$H(\eta'''')-H(\eta'')=2J-(\lambda+h)\ell$, which implies 
$H(\eta'''')<H(\eta'')$.

If $\eta''$ has a single rectangle of pluses, 
this growth mechanism can be continued until $\puno$ is 
obtained proving the statement of the lemma. 
If $\eta''$ has two or more rectangles of pluses, 
this growth mechanism can be continued until two or more interacting 
rectangles are found. In such a case, by performing bootstrap mechanism 
steps and boundary growth of rectangles the $\puno$ configuration 
will be eventually constructed completing the proof of the lemma. 
\qed

\subsection{Proof of Lemma~\ref{t:lan000}}
\par\noindent
Case $h>\lambda$: 
row 11 of table~\ref{tab:heu000}
implies 
that the state $\zero$ is a local minimum of the Hamiltonian, since 
all the possible spin flips have a positive energy cost. 

Case $h<\lambda$: the fact that $\zero$ is a local minimum is proven as above. 
Moreover, 
from row 1 of the tables~\ref{tab:heu000} 
it follows that the state $\muno$ is a local minimum of the Hamiltonian, 
as well.
\qed

\subsection{Proof of Proposition~\ref{thm:out_gammaBCG}}
\par\noindent
The proof of Proposition~\ref{thm:out_gammaBCG} is based on lemmas \ref{lemma:bond_minus_plus}-\ref{lemma:path_0_+}. which are listed at the end of this subsection.
We prove that for every configuration $\eta \not \in \{\muno,\puno\}$, the stability level is strictly smaller than the energy barrier $\Gamma$. For technical reasons, we consider the restricted region $\frac{\lambda}{2}<h<\lambda$, in which the metastable behavior is the same of the region $0<h<\lambda$. 
First of all, given a configuration $\eta \in \mathcal{X}$, we consider the set $\mathcal{C}(\eta) \subseteq \Lambda$ defined as the union of the closed unitary square centered at sites $i$ with the boundary contained in the dual of $\mathbb{Z}^2$ and such that $\eta(i)= +1$. The maximal connected components $C_1, . . . , C_m$, with $m \in \mathbb{N}$, of $\mathcal{C}(\eta)$ are called \emph{clusters of pluses}. We define in the same way the \emph{clusters of minuses} and the \emph{clusters of zeros}. If the boundary of a cluster forms internal right angles, then we call them \emph{convex corners}. Otherwise, we call the other angles \emph{concave angles}. Moreover, we call \emph{convex side} of a cluster the side with both adjacent convex corners. Otherwise, we call the side \emph{concave side}. We observe that each cluster has at least one convex side, since $\Lambda$ is finite and there are zero-boundary conditions.
We partition the set of all configurations in any subsets according to peculiar properties of the contained clusters and we provide the stability level of each of these subsets. In particular, we first analyze the configurations with at least a cluster of pluses and we find their stability level, see Lemmas \ref{lemma:bond_minus_plus}, \ref{lemma:cluster_pluses}, \ref{lemma:subcritical_cluster}, \ref{lemma:supercritical_cluster}, \ref{lemma:supercritical_cluster_no_space}, \ref{lemma:middle_cluster}.
Then, we continue to analyze the remaining configurations by computing the stability level for such configurations that contains only zero and minus spins, see Lemma \ref{lemma:supercritical_cluster}, \ref{lemma:path_0_+}.
In this way, we conclude the proof.
\qed

\begin{lemma}\label{lemma:bond_minus_plus}
Let $\eta$ be a configuration that contains a bond of type $(+,-)$, then there exists a configuration $\eta'$ communicating with $\eta$ with a downhill path. 
\end{lemma}
 \begin{lemma}\label{lemma:cluster_pluses}
Let $\eta$ be a configuration that contains at least a cluster of pluses. If this cluster has a shape different from a rectangle, then $V_\eta<2(\lambda-h)$.
 \end{lemma}
 \begin{lemma}\label{lemma:subcritical_cluster}
If $\eta$ contains either a cluster of pluses with at least a convex side length $l_1<\frac{2J}{\lambda+h}$ or a cluster of minuses with at least a convex side length $l_2<\frac{2J}{\lambda-h}$, then $V_\eta<2J$.
 \end{lemma}
 \begin{lemma}\label{lemma:supercritical_cluster}
 Let $\eta$ be a configuration that contains either a cluster of pluses with at least a side length $l_1>\frac{2J}{\lambda+h}$ at distance strictly greater than two from a minus spin, or a cluster of minuses with at least a side length $l_2>\frac{2J}{\lambda-h}$ at distance strictly greater than two from a plus spin. 
 Then $V_\eta<2J$.
 \end{lemma}

 \begin{lemma}\label{lemma:supercritical_cluster_no_space}
 Let $\eta$ be a configuration that contains a cluster of pluses with at least a side length $l>\frac{2J+\lambda-h}{h}$. Then $V_\eta<5J$.
 \end{lemma}
 \begin{lemma}\label{lemma:middle_cluster}
 Let $\eta$ be a configuration that contains a cluster of pluses with at least a side length $\frac{2J}{\lambda+h}<l<\frac{2J+\lambda-h}{h}$. Then $V_\eta<\Gamma^*$ where $\Gamma^*=\frac{2J^2}{h}$.
 \end{lemma} 
 \begin{lemma}\label{lemma:path_0_+}
The stability level of $\zero$ is strictly smaller than $\Gamma$, i.e., $V_{\zero}<\Gamma$
\end{lemma}
\noindent The proofs of the previous lemmas are in Section \ref{section:lemmas_recurrence}.

\subsection{Proof of Theorem~\ref{thm:recurrence_property_BCG}}
\par\noindent
Let $\Gamma$ as in definition \eqref{eq:energy_barrier}. By applying \cite[Theorem 3.1]{manzo2004essential} with $V^*=\Gamma$ and Proposition \ref{thm:out_gammaBCG}, we get the proof.
\qed

\subsection{Proof of Theorem~\ref{thm:Identification}}
\par\noindent
In this proof, we identify the unique metastable state and we compute the value of the maximal stability level. To do this, we first construct a \emph{reference path} to find an upper bond for the stability level of $\muno$, i.e. $V_\muno \leq \Gamma_m$, and then we give a lower bond of $V_\muno$ such that $V_\muno \geq \Gamma_m$ by using a new procedure based on the computation of the number of bonds in any configuration. In this way, we can conclude the proof.

\subsubsection{Upper bound for $V_{\muno}$} 
\par\noindent
We define the \emph{reference path} as a path from $\muno \to \puno$ consisting in a sequence of configurations with increasing clusters \emph{as close as possible to chopped corner frame} such that $\Phi(\omega)-H(\muno)=\Gamma_m$. Hence, $V_\muno \leq \Gamma_m$. 
We denote by $\sigma^F_{m,n}$ the configuration that contains a chopped corner frame with horizontal side length $m$ and vertical side length $n$. 
We choose the site in one of the four corners of $\Lambda$ and we consider its two nearest neighbors. We flip these three minus spins to zero leaving with an energy cost equals to $3(\lambda-h)$, see Table \ref{tab:heu000}.
Then, we flip the zero in the corner to plus increasing the energy by $4J-(\lambda+h)$. Thus the total energy cost to obtain $\sigma^F_{1,1}$, i.e. to form a chopped corner frame of both side lengths equal to one, is to $4J+2\lambda-4h$.\\
Next, we flip the minus spins at distance smaller than or equal to $\sqrt{2}$ from the zeros, and then we construct a square $2 \times 2$ of pluses. In this way a chopped corner frame of both side lengths two is formed, see Figure \ref{path_minus_plus}
\begin{figure}[H]
\begin{center}
    \includegraphics[scale=0.315]{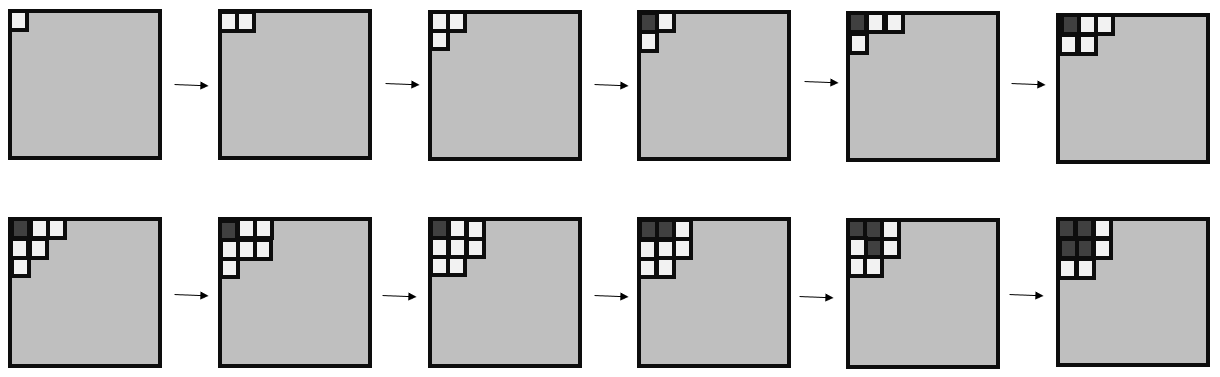}
    \caption{The first part of the reference path, from $\muno$ to $\sigma^F_{2,2}$, a chopped corner frame of both side lengths equal to two.}\label{path_minus_plus}
\end{center}
    \end{figure}
We grow up this chopped corner frame by considering a minus spin at distance one from this frame and from the boundary of $\Lambda$ (this is the effect of the zero-boundary conditions) and by flipping it to zero with an energy cost of $(\lambda-h)$. 
Then, we flip from zero to plus the unique zero with two zero nearest neighbours 
and we repeat these two steps to grow up the chopped corner frame $2 \times 2$ to a chopped corner frame $2 \times 3$. Thus, we obtained $\sigma^F_{2,3}$. 
Next, we grow up the chopped corner frame $2 \times 3$ by considering a minus spin along the longest side at distance one from the frame and the boundary of $\Lambda$ and by flipping it to zero.
Then, we flip from zero to plus the unique zero with two zero nearest neighbours  
and we repeat these two steps until we obtain $\sigma^F_{3,3}$. 
We continue in the same manner by flipping first a minus to zero and then a zero to plus, (see Figure \ref{fig:reference_path}) until the chopped corner frame invades all the lattice $\Lambda$ obtaining the configuration $\puno$. 
\begin{figure}[H]
\begin{center}
    \includegraphics[scale=0.34]{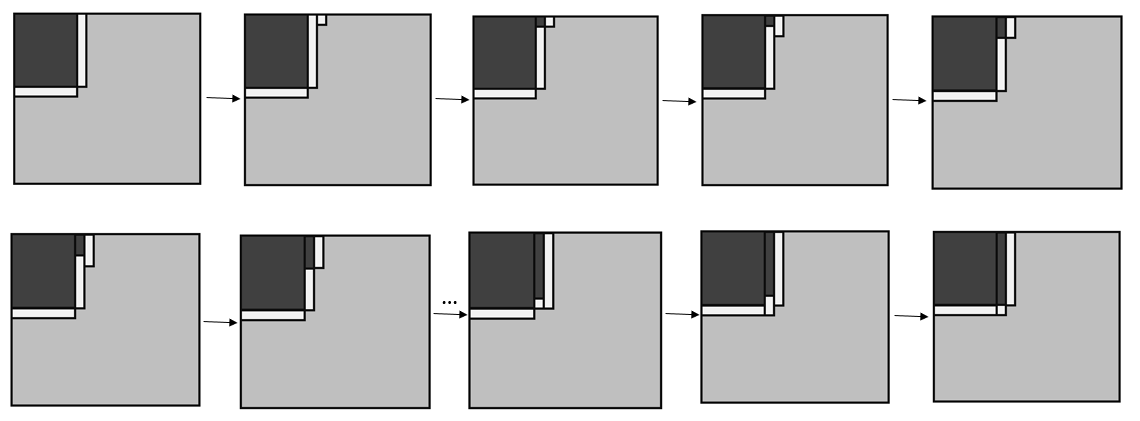}
    \caption{A part of the reference path from $\sigma^F_{n-1,n}$ to $\sigma^F_{n,n}$.}\label{fig:reference_path}
\end{center}
    \end{figure}
    
In the following we compute the communication height of this procedure. First of all, we compute the energy cost between the configuration $\muno$ and a configuration $\sigma^F_{m,n}$. Suppose $m \leq n$, we have
\begin{align}
  H(\sigma^F_{m,n})& =2J(n+m)+4JL-(\lambda+h)nm \notag \\
  & -(\lambda-h)(L^2-nm-n-m)
\end{align}
where $2(m+n)$ is the number of bonds $(0,+)$, $4L$ is the number of bonds $(0,-)$, $mn$ is the number of pluses, and $(L^2-nm-n-m)$ is the number of minuses in $\sigma^F_{m,n}$.
Thus, by equation \eqref{lan000}, we obtain
\begin{align}\label{eq:difference_energy_frame_rectangle}
  H(\sigma^F_{m,n})-H(\muno) & =2J(n+m)+\lambda(n+m) \notag \\
  & -h(2mn+n+m)
\end{align}
In particular,
\begin{align}
  & H(\sigma^F_{n,n-1})-H(\muno)=2J(2n-1)+\lambda(2n-1) \notag \\
  & -h(2n(n-1)+2n-1) \label{eq:difference_energy_frame1} \\
 & H(\sigma^F_{n,n})-H(\muno) =4Jn+2\lambda n-h(2n^2+2n) \label{eq:difference_energy_frame2} \\
  & H(\sigma^F_{n,n+1})-H(\muno) =2J(2n+1)+\lambda(2n+1) \notag \\
  & -h(2n(n+1)+2n+1) \label{eq:difference_energy_frame3}
\end{align}
%
We have that $H(\sigma^F_{n,n-1})>H(\sigma^F_{n,n})>H(\sigma^F_{n,n+1})$ for $n>\frac{2J+\lambda-h}{2h}$, and $H(\sigma^F_{n,n-1})<H(\sigma^F_{n,n})<H(\sigma^F_{n,n+1})$ for $n<\frac{2J+\lambda-h}{2h}$. Thus, the communication height $\Phi(\omega)$ along the reference path $\omega$ is equal to $\Phi(\sigma^F_{l_c,l_c-1},\sigma^F_{l_c,l_c})$, 
where $l_c$ is defined in \eqref{eq:critical_length}.
Starting from $\sigma^F_{l_c, l_c-1}$ to reach the configuration $\sigma^F_{l_c, l_c}$, the maximal height is given by the first three steps and its value is $2J-(\lambda+h)+2(\lambda-h)$. Indeed, the first step is the flip of the minus spin at distance one from the chopped corner frame and from the boundary of $\Lambda$ in to zero. The energy cost of this flip is $(\lambda-h)$. The second step is the flip of the unique zero with two zero nearest neighbours in to plus, and its energy cost is $2J-(\lambda+h)$. The third step is the flip of the minus at distance one from the first flipped minus and at distance two from the boundary of $\Lambda$ in to zero. the energy cost of this last flip is $(\lambda-h)$. The rest of the path to reach $\sigma^F_{l_c, l_c}$ is composed by a sequence of flipping a zero into a plus with the decrease of energy of $(\lambda+h)$ followed by flipping minus into zero with an energy cost of $(\lambda-h)$, thus it is a two-steps downhill path.
Hence, using the equations \eqref{eq:difference_energy_frame1},\eqref{eq:difference_energy_frame2} and \eqref{eq:difference_energy_frame3}, we have
\begin{align}\label{eq:gamma_BCG_esatta}
& \Phi(\muno,\puno)-H(\muno)=\Phi(\sigma_{l_c, l_c-1},\sigma_{l_c,l_c})-H(\muno) \notag \\
& \leq H(\sigma_{l_c, l_c-1})+2J-(\lambda+h)+2(\lambda-h)-H(\muno) \notag \\
& = 4Jl_c +2\lambda l_c-2h2l_c^2-2h
=\Gamma.
\end{align}
We note that $\Gamma>\Gamma^*$, where 
\begin{equation}\label{gamma*}
    \Gamma^*=\frac{2J^2}{h}.
\end{equation}

\subsubsection{Lower bound for $V_{\muno}$} 
\par\noindent
To find the lower bound of $\muno$, we use lemmas \ref{lem:zeros_column_L}-\ref{lem:step_2_minmax}, which are collected at the end of this subsection. We denote with $\mathcal{M}_{n^+}$ the manifold with a fixed number $n^+$ of pluses. Fixed $n^+_c=l_c(l_c-1)$, we define 
\begin{equation}\label{def:critical_conf}
    \sigma_c=\sigma^F_{l_c-1,l_c}
\end{equation}
see Figure \ref{fig:frame_step_selle}. 
\begin{remark}\label{rem:conf_critica} 
We observe that in $\sigma_c$ the smallest rectangle that contains the frame of pluses and zeros has side lengths $l_c$ and $l_c+1$. Moreover, the envelope of this rectangle contains $2(2l_c-1)$ bonds between a plus and a zero, and $2l_c+1$ bonds between a minus and a zero. The other bonds in the envelope are between two spins of the same type. Out of this envelope there are $4L-(2l_c+1)$ bonds between a minus and a zero according to the zero-boundary conditions, and the other bonds are between two spins of the same type.
\end{remark}
To find a lower bound for the stability level of $\muno$, we divide the proof into two main steps: in the first step, we prove that $\sigma_c$ in \eqref{def:critical_conf} is the energy minimizer in the manifold $\mathcal{M}_{n^+_c}$; in the second step we show that a path from $\mathcal{M}_{n^+_c}$ to $\mathcal{M}_{n^+_c+2}$ has minimal communication height  if it starts from $\sigma_c$.
Hence, in Lemma \ref{lem:sigma_c_minimum} we prove that $\sigma_c=\text{argmax}_{\xi \in \mathcal{M}_{n^+_c}}H(\xi)$.
To do this, we use corollary \ref{lem:sequence_of_minuses} and lemmas \ref{lem:region_zero_minus_1}, \ref{lem:region_zero_minus_2}, \ref{lem:zeros_column_L}, \ref{lem:cluster_pluses_different_quasi_square} and \ref{lem:min_crit_quad_no_boundary} to prove that 
if $\eta \in \mathcal{M}_{n^+_c}$ is a configuration that differs from $\sigma_c$ then $\eta \not \in M$, where 
\begin{equation}\label{def:set_minima}
    M=\{\sigma \, | \, H(\sigma)=\min_{\xi \in \mathcal{M}_{n^+_c}} H(\xi)\}.
\end{equation} 

Then, the second step consists in proving that a path from $\mathcal{M}_{n^+_c}$ to $\mathcal{M}_{n^+_c+2}$ with minimal communication height has to start $\sigma_c$. In particular, by applying lemmas \ref{lem:path_sigma_c}-\ref{lem:step_2_minmax}, we show that the path with minimal communication height crosses two peculiar configurations that we call $\tilde \sigma_c$ and $\sigma_s$.
The configuration $\sigma_c$, respectively $\tilde \sigma_c$, differs from $\sigma_s$ as shown in Figure \ref{fig:critical_configuration_saddle} and \ref{fig:frame_step_selle}. The energy of these two configurations are 
\begin{align} 
    & H(\tilde \sigma_c)=H(\sigma_c)+2J-(\lambda+h)+(\lambda-h), \label{eq:energy_sigma_protuberance} \\
    & H(\sigma_s)=H(\sigma_c)+2J-(\lambda+h)+2(\lambda-h). \label{eq:energy_saddle}
\end{align}
\begin{figure}[H]
        \begin{center}
        \includegraphics[scale=0.39]{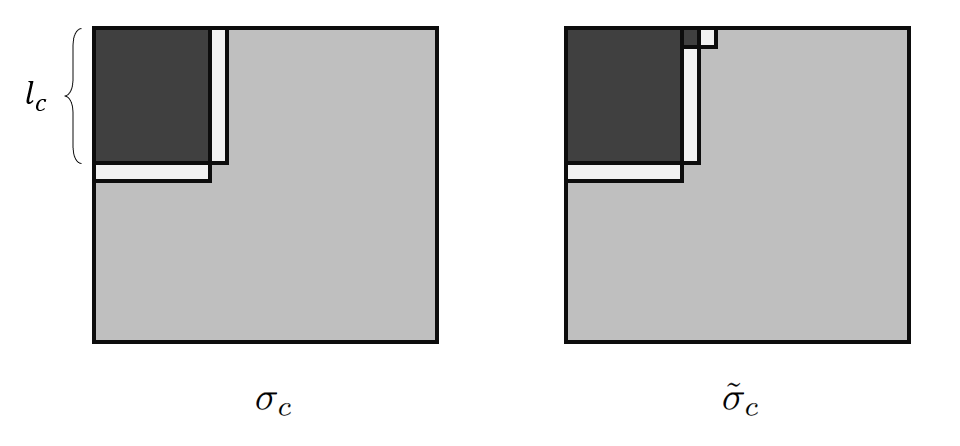}
        \caption{The shapes of the configurations $\sigma_c$ and $\tilde \sigma_c$. We note that the cluster of pluses can be attached in one of the four corners of $\Lambda$ and the protuberance can be attached along one of the two sides of the cluster of pluses. With an abuse of notation, we denote all of these configurations with $\sigma_c$ (or $\tilde \sigma_c$).}
        \label{fig:frame_step_selle}
        \end{center}
    \end{figure}
 Therefore $V_{\muno} \geq H(\sigma_s)=\Gamma$.
Summarizing, $V_{\muno}=\Gamma$ and for any $\eta \in \mathcal{X} \setminus \{\muno,\puno\}$ we have $V_\eta<\Gamma$.
Thus, by \cite[Theorem 2.4]{cn2013}, $\muno$ is the unique metastable state.
\qed

We define \emph{strip of pluses} (resp. \emph{strip of minuses}) the connected subset of a column or a row of $\Lambda$ filled with pluses only (resp. minuses only).
 \begin{corollary}\label{lem:sequence_of_minuses}
 Let $\eta$ be a configuration that contains a strip of minuses with at least a side length $l>\frac{2J}{\lambda-h}$ at distance strictly greater than two from a plus spin. Then $V_\eta<2J$.
 \end{corollary}
\begin{lemma}\label{lem:region_zero_minus_1}
Let $\eta \in \mathcal{M}_{n^+_c}$ be a configuration that contains a rectangle $R$ with side lengths $l_1,l_2 >\lfloor \frac{2J}{\lambda-h} \rfloor+2$ with inside no plus spins. Assume that $\eta_{R\setminus \partial^- R} \neq \muno_{R\setminus \partial^- R}$. If $\eta(x)=-1$ for at least a site $x\in R \setminus \partial^- R$, then there exists a configuration $\eta'\in\mathcal{M}_{n^+_c}$ such that $H(\eta')<H(\eta)$.
\end{lemma}
\begin{lemma}\label{lem:region_zero_minus_2}
Let $\eta \in \mathcal{M}_{n^+_c}$ be a configuration that contains a rectangle $R$ with side lengths $l_1,l_2 >\lfloor \frac{4J}{\lambda-h} \rfloor+2$ with inside no plus spins. Let $S=R \setminus \partial^- R$ and $\eta(x)=0$ for every $x\in S$, then the configuration $\eta'\in \mathcal{M}_{n^+_c}$ such that $\eta'_{\Lambda \setminus S}=\eta_{\Lambda \setminus S}$ and $\eta'_{S}=\muno_{S}$ has $H(\eta')<H(\eta)$.
\end{lemma}
\noindent Recalling the definition of the set $M$ in \eqref{def:set_minima},
\begin{lemma}\label{lem:zeros_column_L}
If $\eta \in \mathcal{M}_{n^+_c}$ is a configuration that contains at least a column (or a row) with only zero spins, then $\eta \not\in M$. 
\end{lemma}
\begin{lemma}\label{lem:cluster_pluses_different_quasi_square}
    Let $\eta \in \mathcal{M}_{n^+_c}$ be a configuration that contains a cluster of pluses with a shape different from a quasi-square with side lengths $l_c$ and $l_c-1$. Then $\eta \not \in M$.
\end{lemma}
 \begin{lemma}\label{lem:min_crit_quad_no_boundary}
If $\eta \in \text{argmax}_{\xi \in \mathcal{M}_{n^+_c}}H(\xi)$, then $\eta_Q=+1_Q$ and $\eta_{{\Lambda \setminus (Q \cap \partial^+ Q)}}=-1_{\Lambda \setminus (Q \cap \partial^+ Q)}$.
 \end{lemma}
\begin{lemma}\label{lem:sigma_c_minimum}
$H(\sigma_c)=\min_{\xi \in \mathcal{M}_{n^+_c}} H(\xi)$.
\end{lemma}
\begin{lemma}\label{lem:path_sigma_c}
    If $\underline{\omega}$ is a path from $\sigma_c$ to $\mathcal{M}_{n^+_c+1}$ such that $\underline{\omega}=(\sigma_c, \eta_1, \eta_2, ... \eta_n)$, $n \geq 1$, with $\eta_n \in \mathcal{M}_{n^+_c+1}$ and $\eta_i \in \mathcal{M}_{n^+_c}$ for every $i=1,...,n-1$, then $\Phi(\underline{\omega}) \geq H(\tilde \sigma_c)$.
\end{lemma}
\begin{lemma}\label{lem:step_1_minmax}
    If $\underline{\omega}$ is a path from $\mathcal{M}_{n^+_c}$ to $\mathcal{M}_{n^+_c+1}$, then $\Phi(\underline{\omega}) \geq H(\tilde \sigma_c)$.
\end{lemma}
\begin{lemma}\label{lem:path_tilde_sigma_c}
    If $\underline{\omega}$ is a path from $\tilde \sigma_c$ to $\mathcal{M}_{n^+_c+2}$ such that $\underline{\omega}=(\tilde \sigma_c, \eta_1, \eta_2, ... \eta_n)$, $n \geq 1$, with $\eta_n \in \mathcal{M}_{n^+_c+2}$ and $\eta_i \in \mathcal{M}_{n^+_c+1}$ for every $i=1,...,n-1$, then $\Phi(\underline{\omega}) \geq H(\sigma_s)$.
\end{lemma}
\noindent We define the set $\mathscr{S}$ as the set of all configurations of $\mathcal{M}_{n^+_c+1}$ such that
\begin{itemize}
    \item[a.] the bonds of type $(+,-)$ are not present;
    \item[b.] the union of the cluster of pluses is composed by only one cluster and its semi-perimeter is equal to $2l_c$;
    \item[c.] the minimal rectangle that contains the cluster of pluses has either side lengths $(l_c,l_c)$ or $(l_c+1,l_c-1)$;
    \item[d.] the envelope of the cluster of pluses has a corner that coincides with a corner of $\Lambda$.
    \item[e.] there is only one strip of minuses in each column and row of $\Lambda$.
\end{itemize}
We note that $\tilde \sigma_c \in \mathscr{S}$.
\begin{lemma}\label{lem:equivalence_tilde_sigma_c}
    Let $\eta \in \mathcal{M}_{n^+_c+1}$ be such that $H(\eta)=H(\tilde \sigma_c)$, then $\eta \in \mathscr{S}$.
\end{lemma}
\begin{lemma}\label{lem:conf_S_different_tilde_sigmac} 
Let $\eta \in \mathscr{S} \setminus \{ \tilde \sigma_c\}$. Every path $\underline{\omega}: \eta \to \muno$ is such that $\Phi(\underline{\omega}) > H(\sigma_s)$.
\end{lemma}
\begin{lemma}\label{lem:step_2_minmax}
  If $\underline{\omega}$ is a path from $\muno$ to $\mathcal{M}_{n^+_c+2}$, then $\Phi(\underline{\omega}) \geq H(\sigma_s)$. 
\end{lemma}
\noindent The proofs of the previous lemmas are in Section \ref{section:lemmas_minmax}.

\subsection{Proof of Theorem~\ref{thm:transition_time_BCG}}
\par\noindent
By applying \cite[Theorem 4.1]{manzo2004essential} with $\eta_0=\muno$ and our value of $\Gamma$, we get the proof.
\qed

\section{Proof of the lemmas of Section~\ref{s:pro-th}}
\label{s:pro-le}
\par\noindent
In Section \ref{section:lemmas_recurrence} we report the proofs of lemmas related to recurrence property, while in Section \ref{section:lemmas_minmax} we gather the proofs of the lemmas related to the computation of the energy barrier. 

\subsection{Proofs of lemmas for the recurrence property} \label{section:lemmas_recurrence}
\par\noindent
\smallskip
\par\noindent
\textit{Proof of Lemma \ref{lemma:bond_minus_plus}.\/}
Using the Table \ref{tab:heu000}, it is possible to reduce the energy of all configurations with a bond $(+,-)$ except the configuration where the plus and the minus are near three pluses and three minuses respectively. In the latter case, we analyze the two columns (or rows) where the plus and the minus belong to until we find a bond different from $(+,-)$. If there is a bond different from $(+,-)$ then the energy of $\eta$ is reducible using the Table \ref{tab:heu000}, otherwise $\eta$ contains two columns composed by all bonds $(+,-)$. In this case, by analyzing the last bond of the two columns in internal boundary of $\Lambda$, we obtain a configuration that is reducible in energy because of the zero-boundary conditions and according to the Table \ref{tab:heu000} at row 5 (by flipping a minus in to zero) or row 9 (by flipping a plus in to zero). 
\qed

In order to prove the following lemmas we define a \emph{local configuration} of a configuration $\eta \in \mathcal{X}$ the rescricted configuration $\eta_{U_x}$, where $x$ is a site of $\Lambda$ and $U_x=\{ y \in \Lambda \, | \, |x-y|=1\}$. See Figure \ref{fig:plus_tiles} for some examples.
\smallskip
\par\noindent
\textit{Proof of Lemma \ref{lemma:cluster_pluses}.\/}
First of all, suppose that $\eta$ does not contain bonds $(+,-)$, otherwise we conclude by lemma \ref{lemma:bond_minus_plus}.
Using Table \ref{tab:heu000}, we find that the only local configurations containing a plus that are not reducible with a flip are as in Figure \ref{fig:plus_tiles}.

\begin{figure}[H]
\begin{center}
\includegraphics[scale=0.8]{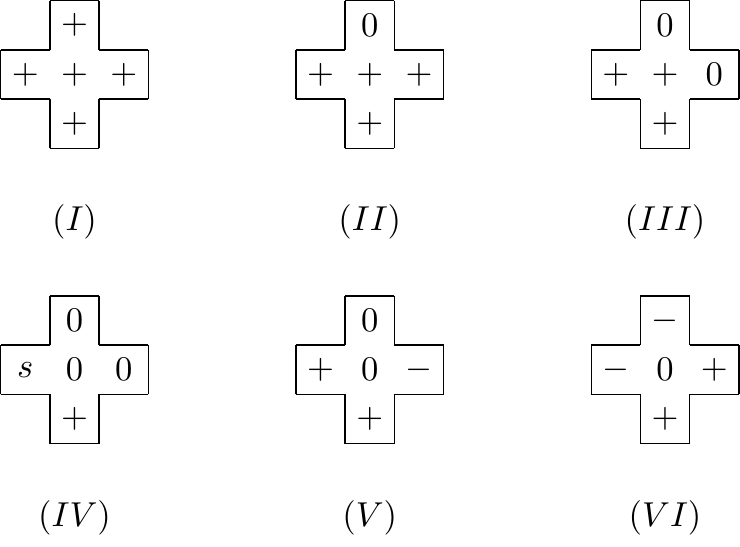}
\end{center}
\caption{Local configurations with at least a plus spin that are not reducible in energy using Table \ref{tab:heu000}. The spin $s$ in picture $(IV)$ takes values in $\{-1,0\}$.}
\label{fig:plus_tiles}
\end{figure}

Let $\eta$ be a configuration that contains at least a local configuration of type $(V)$ in Figure \ref{fig:plus_tiles}. We consider $\eta'$ obtained from $\eta$ by flipping the minus spin to zero and the zero at the center to plus. In this way, we have $H(\eta')<H(\eta)$ and $\Phi(\eta,\eta')-H(\eta)\leq (\lambda-h)$, recalling that $\eta$ does not contain bonds $(+,-)$. \\
Next, suppose that $\eta$ is a configuration that contains at least a local configuration of type $(VI)$. We consider $\eta'$ obtained from $\eta$ in three steps: we flip the two minuses to zero and then we flip the zero at the center to plus. In this way, we have $H(\eta') \leq H(\eta)+2(\lambda-h)-(\lambda+h)$ recalling that $\eta$ does not contain bonds $(+,-)$. For the assumptions $h>\frac{\lambda}{2}$, in particular $h>\frac{\lambda}{3}$, then $H(\eta')<H(\eta)$, and $\Phi(\eta,\eta')-H(\eta)\leq 2(\lambda-h)$ recalling that $\eta$ does not contain bonds $(+,-)$. \\
Follows that a cluster of pluses is composed by only local configurations of types $(I)$, $(II)$, $(III)$ with a plus at the center and only local configurations of type $(IV)$ with a plus in the neighborhood, thus is a rectangle and we have concluded the proof.
\qed

\smallskip
\par\noindent
\textit{Proof of Lemma \ref{lemma:subcritical_cluster}.\/}
We prove the result for the cluster of pluses, the other case is similar. 
Suppose that the cluster of pluses has at least one convex side with length $l_1<\frac{2J}{\lambda+h}$. 
We flip the $l_1$ pluses along the side to zero decreasing in energy with a communication height smaller than or equal to $H(\eta)+(\lambda+h)(l_1-1)<H(\eta)+2J$. Indeed, starting from a corner of the cluster and flipping the first $l_1-1$ pluses, the energy increases by $\lambda+h$ at each flip, since the number of the bonds between two equal spins does not change but a plus is replaced by a zero, see Table \ref{tab:heu000} at row 13. Then, during the $l_1$-th flip, the energy decreases by $2J-(\lambda+h)$, see Table \ref{tab:heu000} at row 12.
\qed

\smallskip
\par\noindent
\textit{Proof of Lemma \ref{lemma:supercritical_cluster}.\/}
We prove the result for the cluster of pluses, the other case is similar. 
Suppose that the cluster of pluses has at least a side with length $l_1>\frac{2J}{\lambda+h}$ at distance strictly greater than two from a minus spin. 
We suppose that there are only zero spins at distance two from the pluses along this side, see Table  \ref{tab:heu000} Then, we consider these $l_1$ zeros and we flip them to plus obtaining $\eta'$ and decreasing in energy. In particular, 
the communication height of the path connecting $\eta$ to $\eta'$ is at most $2J-(\lambda+h)+H(\eta)$ (if the side is convex, otherwise $\Phi(\eta,\eta')=0$ indeed if the side is concave then the energy decreases by $\lambda+h$ see Table \ref{tab:heu000} at row 13), see Table \ref{tab:heu000} at row 12. Indeed the first flip has an energy cost equal to $2J-(\lambda+h)$, since a zero is replaced by a plus and the number of the bonds between two equal spins has decreased by two. The other steps form a downhill path. Thus, denoted by $\underline{\omega}$ this path, we have $\Phi(\underline{\omega})=2J-(\lambda+h)+H(\eta)<2J+H(\eta)$.
\qed


\smallskip
\par\noindent
\textit{Proof of Lemma \ref{lemma:supercritical_cluster_no_space}.\/}
Let $\eta$ be a configuration as in the assumption. Suppose that $\eta$ does not contain bonds of type $(+,-)$ otherwise we conclude applying Lemma \ref{lemma:bond_minus_plus}. Moreover, the cluster of pluses is a rectangle otherwise the statement is proven by Lemma \ref{lemma:cluster_pluses}. 
We consider a configuration $\eta'$ obtained from $\eta$ in the following way. All minuses at distance $\sqrt{2}$ and $2$ from the side of the rectangle with length $l>\frac{2J+\lambda-h}{h}$ in $\eta$ are replaced by zeros. Moreover, all zeros at distance one from the same side are replaced by pluses, see Figure \ref{fig:recurrence_5J}. Next, we construct a path $\eta \to \eta'$ with $\Phi(\underline{\omega})-H(\eta)<5J$ and we show that $H(\eta')<H(\eta)$. In the worst case scenario, all spins at distance $\sqrt{2}$ and $2$ from the rectangle are minuses. Thus, in particular we start flipping the two minuses at distance $\sqrt{2}$ from the side of the rectangle, and the energy increases by $2(\lambda-h)$. Next, we consider one of the $l$ minuses at distance two from the considered side of the rectangle, and we flip it to zero. Then, we flip the nearest zero to plus. Starting from a minus at distance one from the minus considered before, we iterate these two steps ($-1 \to 0$ and $0 \to +1$) for $l-1$ times obtaining $\eta'$ such that $H(\eta')<H(\eta)$. Indeed, the first flip of the minus to zero has an energy cost of $2J+(\lambda-h)$ and the first flip of the zero to plus has an energy cost of $2J-(\lambda+h)$, see Table \ref{tab:heu000} at row 2 and 12 respectively. The rest of the steps has an energy cost of $\lambda-h$ when we flip a minus to zero and $-(\lambda+h)$ when we flip a zero to plus. Thus, we have $H(\eta') \leq H(\eta)+4J+2(\lambda-h)-2hl<H(\eta)$ since $l>\frac{2J+\lambda-h}{h}$, and the communication height along this path is $2(\lambda-h)+[2J+(\lambda-h)]+[2J-(\lambda+h)]+(\lambda-h)=4J+3\lambda-5h<5J$ since we chose $J>>\lambda>h$.

 \begin{figure}[H]
 \begin{center}
     \includegraphics[scale=0.6]{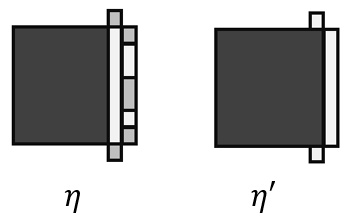}
     \caption{On the left, in dark gray, an example of cluster of pluses of $\eta$ with some minus spins at distance $\sqrt{2}$ and $2$. On the right, the evolution of this cluster in $\eta'$: all minuses at distance $\sqrt{2}$ and $2$ from the cluster are replaced by zeros, and all zeros at distance one are replaced by pluses.} \label{fig:recurrence_5J}
 \end{center}
     \end{figure}
\qed


\smallskip
\par\noindent
\textit{Proof of Lemma \ref{lemma:middle_cluster}.\/}
We observe that if $\eta$ contains a cluster of pluses with at least a side length $l>\frac{2J}{\lambda+h}$ at distance strictly greater than two from a minus spin, then the proof is concluded by Lemma \ref{lemma:supercritical_cluster}. Thus, suppose that there are some minuses at distance $d$ smaller than or equal to two from the cluster of pluses. In particular $\sqrt{2} \leq d \leq 2$, otherwise there is a bond of type $(+,-)$ and we conclude the proof by Lemma \ref{lemma:bond_minus_plus}. Moreover, the cluster of pluses is a rectangle, otherwise the proof is over by Lemma \ref{lemma:cluster_pluses}. We observe that the rectangle of pluses has both side lengths in $(\frac{2J}{\lambda+h},\frac{2J+\lambda-h}{h})$, otherwise we conclude applying Lemma \ref{lemma:subcritical_cluster} or Lemma \ref{lemma:supercritical_cluster_no_space}. 
Denote by $l_+=\lceil\frac{2J}{\lambda+h}\rceil$ and $l_F=\lfloor\frac{2J+\lambda-h}{h}\rfloor$, moreover we indicate by $\tilde l=\lfloor\frac{J+\lambda+h}{h}\rfloor$.
Next, we construct a path $\underline{\omega}$ from $\eta$ to $\eta'$, where $\eta'$ is a configuration such that $H(\eta')<H(\eta)$ and $\Phi(\underline{\omega})-H(\eta)<\Gamma$. In order to find $\eta'$, we distinguish two cases.
Let $m,k$ be the two side lengths of the rectangle of pluses and we suppose $k \geq m$, then we have:
\begin{itemize}
    \item[1.] both sides have length strictly greater then $\tilde l$, that is $k,m \in [\tilde l+1, l_F]$.
    \item[2.] at least one of two side lengths is smaller than $\tilde l$, that is $m \in [l_+, \tilde l]$.
\end{itemize}
In the first case, we obtain $\eta'$ growing the rectangle of pluses as in proof of Lemma \ref{lemma:supercritical_cluster_no_space}. In particular, we grow the side of the rectangle with length $k$ for $l_F-k$ times, that is the rectangle grows up until it reaches the longer side length $l_F$. 
We observe that to grow the side of length $k$, we have to add $m$ pluses along the side of length $m$, see Figure \ref{fig:recurrence_1}. We call $\tilde\eta$ this configuration. Along this first part $\eta \to  \tilde \eta$ of the path $\eta \to  \eta'$, the energy increases, because the rectangle is not supercritical. Then, we will grow up a supercritical rectangle until we obtain $\eta'$ with $H(\eta')<H(\eta)$. Along this last part of path the energy decreases because it is a two-steps downhill path, so the communication height between $\eta$ and $\eta'$ is the same between $\eta$ and $\tilde \eta$. Then, as proof of Lemma \ref{lemma:supercritical_cluster_no_space}, we have
\begin{equation}
    \Delta H\text{(side growth of length $m$)} \leq 4J+2(\lambda-h)-2hm
\end{equation}
Thus, we obtain
\begin{align}
& \Delta H \text{(total growth)} \leq ( l_F-k)\Delta H \text{(growth side of length $m$)} \notag \\
& \leq \Big( \frac{2J+\lambda-h}{h}-k\Big) (4J+2(\lambda-h)-2hm).
\end{align}
To find an upper bound for the communication height, we have to sum the energy difference from the rectangle with longer side length $k$ to $l_F$ with the energy cost to reach the rectangle with side length $l_F+1$. In particular, we conclude finding the following upper bound 
\begin{align}
  \Phi(\eta,\eta')-H(\eta) & \leq \sum_{j=k}^{l_F} \Delta H \text{(growth side of length $m$)} \notag \\
  & +(4J+3\lambda-5h) \notag \\
  & \leq (4J+2(\lambda-h)-2h m) (l_F-k+1) \notag \\
  &+(4J+3\lambda-5h) \notag \\
  & \leq (4J+2(\lambda-h)-2h (\tilde l+1)) ( l_F-\tilde l) \notag \\
  &+(4J+3\lambda-5h) \notag \\
  & < 
  \Gamma^*.
\end{align}
where the second inequality follows from $k,m \geq \tilde l$, and the last one follows from $\tilde l=\lfloor\frac{J+\lambda+h}{h}\rfloor$ and $J>>\lambda>h$. 

 \begin{figure}[H]
 \begin{center}
     \includegraphics[scale=0.5]{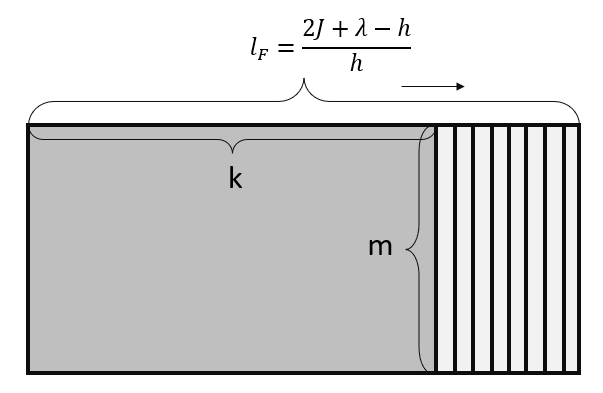}
     \caption{The rectangular cluster of pluses with side lengths $k$ and $m$ grows the side with length $m$ for $l_F-k$ times. In this way we obtain a configuration containing a rectangular cluster of pluses with side lengths $l_F$ and $m$.}\label{fig:recurrence_1}
 \end{center}
     \end{figure}

In the second case, we obtain $\eta'$ shrinking the rectangle of pluses as in proof of Lemma \ref{lemma:subcritical_cluster}. In particular, we cut the side of the rectangle with length $m$ until the cluster of pluses is replaced by a cluster of zeros, see Figure \ref{fig:recurrence_2}. 
First of all, we prove that $H(\eta')<H(\eta)$. We observe that
\begin{align}\label{eq:k_min}
    k \leq l_F< \frac{2Jm}{(\lambda+h)m-2J},
\end{align}
where the second inequality is due to $h>\frac{\lambda}{2}$. Then, by \eqref{eq:k_min} we have
\begin{align}
    H(\eta)-H(\eta') & =2J(k+m)-(\lambda+h)km>0.
\end{align}

To find an upper bound for the communication height, first of all we compute the energy to cut a side of the rectangle and the communication height along this part of the path $\underline{\omega}$. 
For the first $k-1$ times, we have
\begin{equation}
    \Delta H \text{(shrink side of length $m$)}=(\lambda+h)m-2J
\end{equation} 
And $\Phi(\underline{\omega})-H(\eta)=(\lambda+h)m$. Indeed, when we cut $k-1$ sides of length $m$, we obtain a configuration with a rectangle $1 \times m$, so the path toward $\eta'$ is a downhill path. 
Thus, 
\begin{align}
  \Phi(\eta,\eta')-H(\eta) & \leq \sum_{j=1}^{k-2} \Delta H \text{(shrink side of length $m$)} \notag \\
  & +(\lambda+h)m \notag \\
  & = [(\lambda+h)m-2J] (k-2)+(\lambda+h)m \notag \\
  & < [(\lambda+h)\tilde l-2J] (l_F-2)+(\lambda+h) \tilde l \notag \\
  & < \frac{2J^2}{h}=\Gamma^*.
\end{align}
where for the first inequality we used $m \leq \tilde l$ and $k\leq l_F$. The second inequality follows from the values of $\tilde l$, $l_F$ and the assumption $h>\frac{\lambda}{2}$, $J >>\lambda>h$. 

\begin{figure}[H]
\begin{center}
    \includegraphics[scale=0.6]{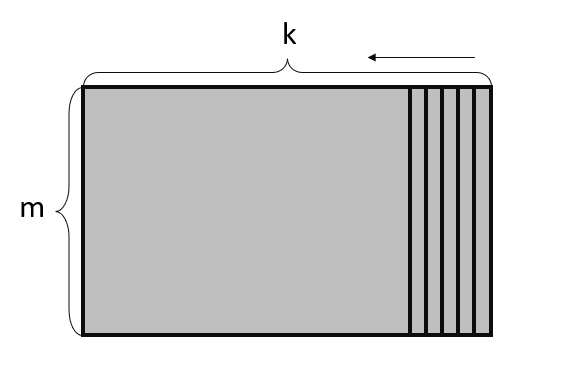}
    \caption{The rectangular cluster of pluses with side lengths $k$ and $m$ shrinks until it is totally replaced by a rectangular cluster of zeros with the same size.}\label{fig:recurrence_2}
\end{center}
    \end{figure}

\qed

\smallskip
\par\noindent
\textit{Proof of Lemma~\ref{lemma:path_0_+}.\/}
To prove the result, we provide a path from $\zero$ to $\puno$.
We define our path $\underline{\omega}: \zero \mapsto \puno$ as a sequence of configurations from $\zero$ to $\puno$ with increasing clusters \emph{as close as possible to quasi-square}, see Figure \ref{fig:path_+_0}. 
\begin{figure}[H]
\begin{center}
    \includegraphics[scale=0.5]{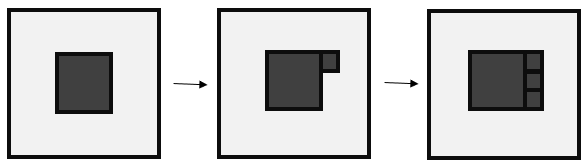}
    \caption{A part of the path $\underline{\omega}: \zero \mapsto \puno$. The white part represents the region with zero spins, the dark gray region is the cluster of pluses. We remark that the first flip from zero to plus can occur at any site of $\Lambda$ with the same probability, this is the case of the homogeneous nucleation.}\label{fig:path_+_0}
\end{center}
    \end{figure}
We construct a path in which at each step we flip one spin from zero to plus. We flip the spin at the origin and then we add clockwise three square units to obtain the first square with side length $l=2$. Then we flip the zero spins on the top of the square $2 \times 2$, adding consecutive square units until we obtain a quasi-square $2 \times 3$. Next we flip the zero spins along the longest side to obtain a square $3 \times 3$. We go on in the same manner flipping consecutive zero spins at distance one to the cluster of pluses. We iterate this nucleation process until the quasi-square takes up all the space $\Lambda$. In the following we compute the communication height of this procedure. First of all, we compute the energy cost between the configuration $\zero$ and a configuration with a rectangular cluster of pluses with side lengths $m$ and $n$, called $\sigma_{m,n}$,
\begin{align}\label{eq:difference_energy_rectangle}
    H(\sigma_{m,n})-H(\zero)=2J(n+m)-(\lambda+h)mn
\end{align}
where $2(m+n)$ is the number of bonds $(0,+)$ and $mn$ is the number of pluses in $\sigma_{m,n}$. The equation \eqref{eq:difference_energy_rectangle} attains the maximum for $(m,n)=\Big(\frac{2J}{\lambda+h},\frac{2J}{\lambda+h}\Big)$, that corresponds to a configuration with a square of pluses with side length $\tilde n=\lfloor\frac{2J}{\lambda+h}\rfloor +1$. Starting from $\sigma_{\tilde n, \tilde n}$ to reach the configuration $\sigma_{\tilde n+1, \tilde n}$, the energy cost is given by the first step and its value is $2J-(\lambda+h)$, see Table \ref{tab:heu000} at row 12, the rest of the path is a downhill path. Thus, recalling that $H(\zero)=0$, using the value of $\tilde n$ and the assumption $\lambda>h$, we have
\begin{align}
& \Phi(\zero,\puno)-H(\zero) \leq \Phi(\sigma_{\tilde n, \tilde n},\sigma_{\tilde n+1,\tilde n})-H(\zero) \notag \\
&=H(\sigma_{\tilde n, \tilde n})+2J-(\lambda+h)-H(\zero) \notag \\
&=4J\tilde n-(\lambda+h)\tilde n^2+2J-(\lambda+h) \notag \\
&=\frac{4J^2}{\lambda+h}+2J-2(\lambda+h) <\frac{2J^2}{h}<\Gamma.
\end{align}
\qed

\subsection{Proofs of lemmas for the energy barrier}\label{section:lemmas_minmax}
\par\noindent
 \begin{proof}[Proof of Corollary \ref{lem:sequence_of_minuses}]
If a configuration contains a strip of minuses as in the assumptions, then there exists a cluster of minuses containing this strip with at least a side length $l>\frac{2J}{\lambda-h}$ at distance strictly greater than two from a plus spin, then we conclude by applying Lemma \ref{lemma:supercritical_cluster}. 
 \end{proof}
\begin{proof}[Proof of Lemma \ref{lem:region_zero_minus_1}]
Let $\eta$ be a configuration as in the assumptions. 
We distinguish two cases: (i) $\eta$ contains at least a cluster of minuses with shape different from a rectangle, (ii) $\eta$ contains only cluster of minuses with rectangular shape. In the first case, there is at least a zero spin with two minus spins at distance one, then we find $\eta'$ by using Table \ref{tab:heu000} at row 4 (by flipping this zero in to a minus). In the second case, we find $\eta'$ by applying either Lemma \ref{lemma:subcritical_cluster} or Lemma \ref{lemma:supercritical_cluster}, according to the side length of the cluster of minuses.
\end{proof}
\begin{proof}[Proof of Lemma \ref{lem:region_zero_minus_2}]
Consider $\eta$ and $\eta'$ as in the assumption. 
The energy difference between $\eta$ and $\eta'$ is given by 
\begin{align}
     H(\eta')-H(\eta) & =2J(l_1-2+l_2-2) \notag \\
     & -(\lambda-h)(l_1-2)(l_2-2) \notag \\
     & <2J\frac{8J}{\lambda-h}-(\lambda-h)\Big(\frac{4J}{\lambda-h}\Big)^2=0.
 \end{align}
To compute the communication height $\Phi(\eta,\eta')$, we argue as in proof of Lemma \ref{lemma:path_0_+}. Indeed the computation of $\Phi(\eta,\eta')$ is similar to one of $\Phi(\zero,\puno$, hence it is strictly smaller than $\Gamma^*$.
\end{proof}

\begin{proof}[Proof of Lemma \ref{lem:zeros_column_L}]
Let $\eta$ be a configuration as in the assumption and we suppose by contradiction that $\eta \in M$. 
First of all, we observe that if $\eta$ contains at least one of the local configurations in Figure \ref{fig:zero_distance_plus} (or one of their rotations), then there exist $\eta' \in \mathcal{M}_{n^+_c}$ such that $H(\eta')<H(\eta)$ by using Table \ref{tab:heu000}, thus $\eta \not \in M$.
\begin{figure}[H]
\begin{center}
    \includegraphics[scale=0.8]{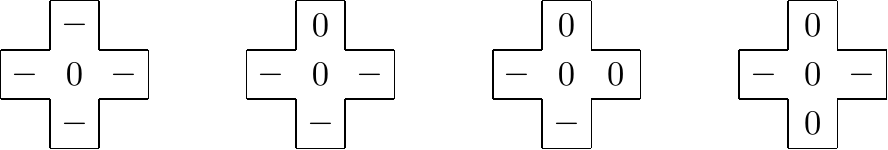}
   \caption{If $\eta$ contains one of these local configuration, then it reducible in energy by flipping the zero in the center in to minus, see Table \ref{tab:heu000}.}\label{fig:zero_distance_plus}
\end{center}
   \end{figure}
From now on, we suppose that $\eta$ does not contain the previous local configurations in Figure \ref{fig:zero_distance_plus}.
For the assumption, $\eta$ contains at least a column (or a row) with only zero spins, then $\eta$ contains at least one of the configurations in Figure \ref{fig:column_zeros}.
\begin{figure}[H]
\begin{center}
    \includegraphics[scale=0.8]{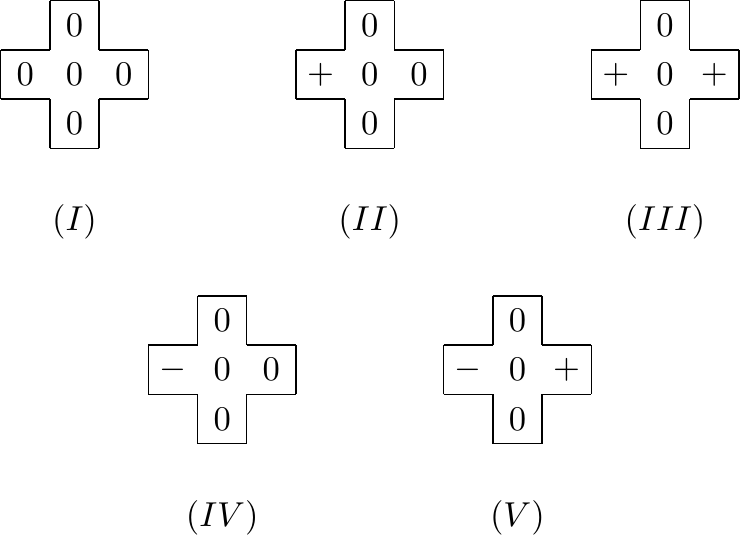}
   \caption{Local configurations with the center site along a column filled by only zero spins.}\label{fig:column_zeros}
\end{center}
   \end{figure}
We observe that $\eta$ does not contain only local configurations of type $(I)$ among those in Figure \ref{fig:column_zeros}, indeed $\eta \in \mathcal{M}_{n^+_c}$. Moreover, we show that if $\eta$ contains only local configurations of type $(I)$ and $(II)$ among those in Figure \ref{fig:column_zeros}, then $\eta \not \in M$. Indeed, in this case $\eta$ does not contain minus spins and by \cite{alonso1996three} we have $H(\eta)\geq H(\xi)$ where $\xi$ is the configuration with a quasi-square of pluses in a sea of zeros, and for $L$ large enough 
we have
\begin{align}
    & H(\xi)= 4J(2l_c-1)-(\lambda+h)l_c(l_c-1),\\
    & H(\sigma_c)=4J(2l_c-1)-(\lambda+h)l_c(l_c-1) +4JL \notag \\
    & \qquad \,\,\,\, -(\lambda-h)(L^2-l_c(l_c-1)-(2l_c-1)),
\end{align}
that is $H(\xi)>H(\sigma_c)$, and $\eta \not \in M$.
With the same argument, we may state that if $\eta$ contains only local configurations of type $(I)$, $(II)$ and $(III)$ among those in Figure \ref{fig:column_zeros}, then $\eta \not \in M$.

Thus, we suppose that $\eta$ contains at least a local configurations of type $(IV)$ or $(V)$ and 
we start to analyze the two columns (or rows) that contain the pair $(-,0)$ until we find a pair $(\eta(x),0)$ such that $\eta(x) \neq -1$. 
First of all, we observe that if the strip of minuses in the first column has a length smaller than $\frac{2J}{\lambda-h}$, then there exists a configuration $\eta' \in \mathcal{M}_{n^+_c}$ with $H(\eta')<H(\eta)$ by Lemma \ref{lemma:subcritical_cluster}, and so $\eta \not \in M$.
Moreover, if $\eta(x)=+1$ then we find $\eta' \in \mathcal{M}_{n^+_c}$ such that $H(\eta')<H(\eta)$ by using Table \ref{tab:heu000}, and also in this case $\eta \not \in M$. 
Thus, the unique possible pair $(\eta(x),0)$ is $(0,0)$.
In this case, there is a plus spin at distance two from the strip of minuses, otherwise $\eta$ satisfies the assumptions of Corollary \ref{lem:sequence_of_minuses} and so $\eta \not \in M$, see Figure \ref{fig:column}. 

\begin{figure}[H]
\begin{center}
    \includegraphics[scale=0.9]{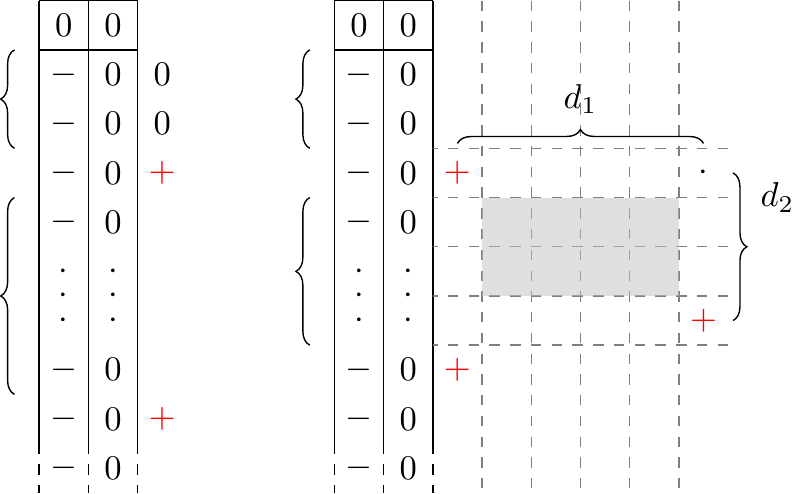}
    \caption{Neighborhood of the column with only zeros spins with attached a strips of minuses.}\label{fig:column}
\end{center}
    \end{figure}

Moreover, for every configuration that contains a pair of two consecutive columns filled by minuses and zeros, there are some plus spins that split the strips of minuses in parts with length smaller than $\frac{2J}{\lambda-h}$, see Figure \ref{fig:column} for an example, otherwise we can reduce the energy of $\eta$ by applying Corollary \ref{lem:sequence_of_minuses}, and so $\eta \not \in M$.
This implies that the distance between two pluses at distance two from the strip of minuses is smaller than $\frac{2J}{\lambda-h}$, see Figure \ref{fig:column}.

Starting from the pair $(0,0)$, we focus on the first plus at distance two from the column of minuses and we consider the plus in the nearest column, see Figure \ref{fig:column}. We observe that the region between these pluses contains only zero and minus spins for construction. In the following, we will prove that this region is a rectangle with both side lengths smaller than $\lfloor\frac{4J}{\lambda-h}\rfloor+2$. 
Indeed, if this region is a rectangle with side length greater than $\lfloor\frac{4J}{\lambda-h}\rfloor+2$ and it contains only zero spins, then we can apply Lemma \ref{lem:region_zero_minus_2} and so $\eta \not \in M$. However, if this region contains some minus spin then we may apply Lemma \ref{lem:region_zero_minus_1}, indeed the assumption and $\eta_{R\setminus \partial^+ R} \neq \muno_{R\setminus \partial^+ R}$ is satisfied otherwise $\eta$ contains the last local configuration in Figure \ref{fig:zero_distance_plus}. 
Hence, the considered region is a rectangle with both side lengths smaller than $\lfloor\frac{4J}{\lambda-h}\rfloor+2$.
Let $d_1$ be the distance between the two columns containing the two plus spins, $d_2$ be the distance between the two rows containing the two plus spins, then $d_1,d_2 < \lfloor\frac{4J}{\lambda-h}\rfloor+2$, see Figure \ref{fig:column}. 
Thus, the Euclidean distance between the two pluses has to be smaller than $\sqrt{2} \Big(\lfloor\frac{4J}{\lambda-h}\rfloor+2 \Big)$. So, we can compute the maximal size of the minimal rectangle containing all plus spins. Indeed the diagonal of this rectangle is $n^+_c \sqrt{2} \Big(\lfloor\frac{4J}{\lambda-h}\rfloor+2 \Big)$ and its side lengths are smaller than $l_R=n^+_c \Big(\lfloor\frac{4J}{\lambda-h}\rfloor+2 \Big)$.

Let $$\tilde \partial^+ R=\partial^+ R \cup \{ x \in \Lambda \setminus R \, : \, |x-y|=\sqrt{2} \,\,\,\,\, \forall \, y \in R\},$$ 
the region $\Lambda \setminus (R\cup \tilde \partial^+ R)$ can be composed by two, three or four rectangles that circumscribing $R$, see Figure \ref{fig:R_Lambda}. We consider the rectangle $R_M$ with maximal area among them and we prove that it has side lengths strictly greater than $\lfloor \frac{4J}{\lambda-h} \rfloor+2$. The maximal rectangle contained in $\Lambda \setminus (R\cup \tilde \partial^+ R)$ has side lengths $(L,x)$ with $x \geq \frac{L}{2}-n^+_c \Big(\lfloor\frac{4J}{\lambda-h}\rfloor+2 \Big)-1$. In particular, we have $L, x>\lfloor \frac{4J}{\lambda-h} \rfloor+2$, since $L>\Big(\frac{2J}{\lambda-h}\Big)^3$. Therefore, for every position of $R$ in $\Lambda$, there is a rectangle $R_M$ that contains only minus spins, otherwise it satisfies the assumption of either Lemma \ref{lem:region_zero_minus_1} or Lemma \ref{lem:region_zero_minus_2}, and so $\eta \not \in M$.
Moreover, there is a strip of minus with length $y> \lfloor\frac{2J}{\lambda-h}\rfloor$, see Figure \ref{fig:R_Lambda}, attached to $R_M$. Thus, the rectangle $S_M$ attached to $R_M$, see Figure \ref{fig:R_Lambda}, is filled by only minus spins, otherwise we can apply Corollary \ref{lem:sequence_of_minuses} and $\eta\not \in M$. Follows that the column with length $L$ filled by only zero spins is not in $\Lambda \setminus (R \cup \tilde \partial^+ R)$, then it is in $\tilde \partial^+ R \cup R$. However, every column (and row) in $\tilde \partial^+ R \cup R$ has length strictly smaller than $L$, thus it is a contradiction.
We can conclude $\eta \not \in M$.
\end{proof}
%
    %
\begin{proof}[Proof of Lemma \ref{lem:cluster_pluses_different_quasi_square}]
Let $\eta$ be a configuration as in the assumption and suppose by contradiction that $\eta \in M$. 
Let $n^0_{\eta}$ be the number of the zero spins in $\eta$.
We first show that if a column or a row 
contains only plus and zero spins, then $\eta \not \in M$. 
Suppose that $\eta$ contains a row $r$ with only plus and zero spins 
and we consider the maximal sequence of $N>0$ consecutive columns that intersects $r$ without plus spins. This set of consecutive columns forms a rectangle $R_{L,N}$ and we note that $N>\lfloor\frac{2J}{\lambda-h}\rfloor+2$, indeed $\frac{L}{n^+_c}> \lfloor\frac{2J}{\lambda-h}\rfloor+2$ see Condition \ref{c:mod007}. If one of them contains only zero spins, then $\eta \not \in M$ by Lemma \ref{lem:zeros_column_L}. Thus, we may apply Lemma \ref{lem:region_zero_minus_1} and we obtain $\eta\not \in M$.
\begin{figure}[H]
\begin{center}
        \includegraphics[scale=0.5]{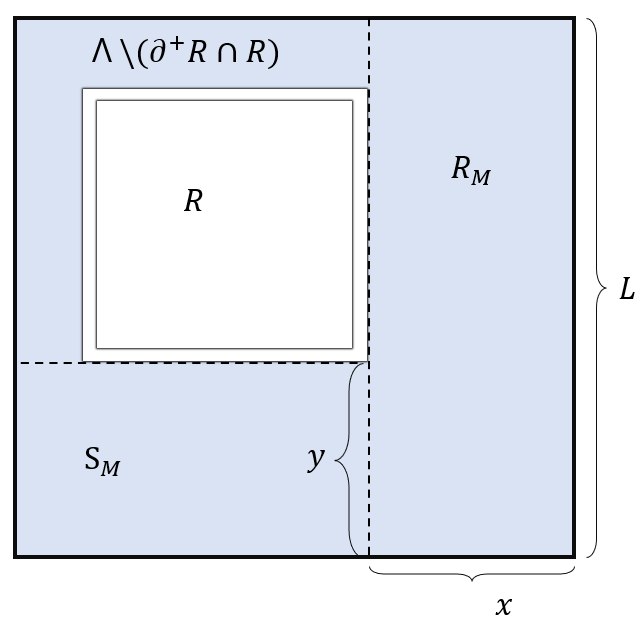}
        \caption{The minimal value of $x$ and $y$ is 
        $\frac{L}{2}-n^+_c \Big(\lfloor\frac{4J}{\lambda-h}\rfloor+2 \Big)-1$,
        when $R$ centered in the middle of $\Lambda$. In each case $L, x>\lfloor \frac{4J}{\lambda-h} \rfloor+2$.}
        \label{fig:R_Lambda}
\end{center}
    \end{figure}
\begin{figure}[H]
\begin{center}
        \includegraphics[scale=0.5]{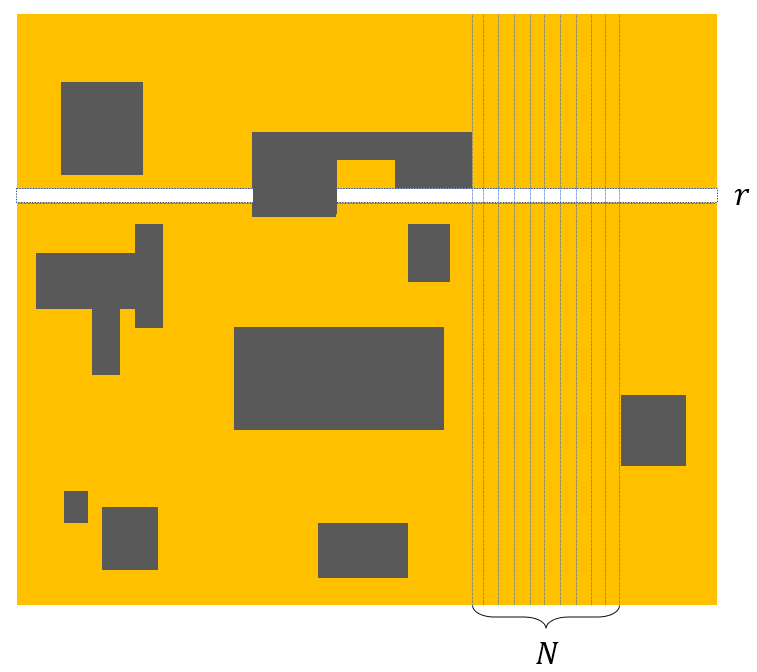}
        \caption{In the picture, the cluster of pluses are in dark gray. The white region indicates the zero region, while the yellow region contains a mixture of zeros and minuses. The set of the $N$ column without pluses that intersect $r$ is the rectangle $R_{L,N}.$}
        \label{fig:RS}
\end{center}
    \end{figure}
Follows that, for each column and row 
that contains at least plus, one of the following conditions holds:
\begin{itemize}
    \item[a)]  there are two bonds $(+,0)$ and at least two bonds $(-,0)$. No bond $(+,-)$ is present. In this case the energy contribution is at least $4J$, by the definition of the hamiltonian function \ref{mod005}, and we denote by $\alpha_1$ the number of these columns and rows. 
    See Figure \ref{fig:RS}.
    \item[b)] there are a bond $(+,-)$, at least a bond $(-,0)$ and a bond $(+,x)$ where $x \in \{-1,0\}$. No more than one bond $(+,0)$ is present. The energy contribution is at least $6J$ and we denote by $\alpha_2$ the number of these columns and rows. See Figure \ref{fig:RS}.
    \item[c)] there are either at least four bonds $(+,0)$ and at least two bonds $(-,0)$, or at least a bond $(+,-)$, at least a bond $(-,0)$ and more than one bond $(+,x)$ where $x \in \{-1,0\}$. The energy contribution is at least $6J$ and we denote by $\alpha_3$ the number of these columns and rows. See Figure \ref{fig:RS}
\end{itemize}
Moreover, we observe that in $\eta$ there are no column filled with only zero spins, otherwise we apply Corollary \ref{lem:sequence_of_minuses}. Then, 
the energy contribution along every column and every row is at least $2J$ according to the zero-boundary conditions. We denote by $\alpha_4$ 
the number of these columns and rows.

We note that $\sum_{i=1}^4\alpha_i=2L$, since we found a partition of the set of columns and rows in $\Lambda$ according to the presence or the absence of pluses.
Moreover, we observe that by \cite{alonso1996three} a cluster of pluses with fixed area $n^+_c=l_c(l_c-1)$ has perimeter $p \geq 2(2l_c-1)$. In particular, the cluster with area $n^+_c$ has minimal perimeter if and only if it is a quasi-square with semi-perimeter $2l_c-1$. In our case, the cluster of pluses has area $n^+_c$ and a shape different from a quasi-square for assumption, then its semi-perimeter is strictly greater than $2l_c-1$. We note that the semi-perimeter of such cluster coincide with the number of columns and rows with a plus, that is $\sum_{i=1}^3\alpha_i \geq 2l_c-1$.

Let $n^-_{\eta}$ be the number of minuses in $\eta$, we can write the energy function of $\eta$ as
\begin{align}\label{eq:energy_columns_rows}
    H(\eta) & \geq 4J\alpha_1+6J\alpha_2+6J \alpha_3+2J\alpha_4 \notag \\
    &-n^+_c(\lambda+h)-n^-_{\eta}(\lambda-h),
\end{align}
and we have
\begin{align}\label{eq:energy_eta_no_quasi_quad}
     H(\eta) &\geq 4J(\alpha_1+\alpha_2+\alpha_3)+2J(\alpha_2+\alpha_3) \notag \\
    & +2J(2L-\alpha_1-\alpha_2-\alpha_3)-n^+_c(\lambda+h) \notag \\
    &-n^-_{\eta}(\lambda-h) \notag \\
    & = 2J (\alpha_1+\alpha_2+\alpha_3)+2J(\alpha_2+\alpha_3+2L) \notag \\
    &-n^+_c(\lambda+h)-n^-_{\eta}(\lambda-h).
\end{align}
Recalling the remark \ref{rem:conf_critica}, we rewrite the energy of $\sigma_c$ and in the following we compare $H(\sigma_c)$ with the energy of $\eta$ in \eqref{eq:energy_eta_no_quasi_quad}.  
\begin{align}\label{eq:energy_critical_conf}
H(\sigma_c) &= 4J (2l_c-1)+2J(2L-(2l_c-1)) \notag \\
& -n^+_c(\lambda+h)-n^-_{\sigma}(\lambda-h) \notag \\
& =2J (2l_c-1)+4JL-n^+_c(\lambda+h)-n^-_{\sigma_c}(\lambda-h).
\end{align}
We distinguish two cases according to the number of zeros in $\eta$:
\begin{itemize}
    \item[1.] $k > 2l_c-1=n^0_{\sigma_c}$;
    \item[2.] $k\leq 2l_c-1=n^0_{\sigma_c}$.
\end{itemize}
In the first case, by \eqref{eq:energy_eta_no_quasi_quad} and recalling $\sum_{i=1}^3\alpha_i \geq 2l_c-1$, we obtain
\begin{align}\label{eq:eta_no_quasi_quad_sigma_c}
   &H(\eta)-H(\sigma_c) \geq 2J(\alpha_2+\alpha_3) +(n^-_{\sigma_c}-n^-_{\eta})(\lambda-h) \notag \\
   &=2J(\alpha_2+\alpha_3) +(n^0_{\eta}-n^0_{\sigma_c})(\lambda-h) >0
\end{align}
since $\alpha_2,\alpha_3\geq 0$ and $n^0_{\eta}-n^0_{\sigma_c}>0$.

In the second case, we observe that $\eta$ has to contain a bond $(+,-)$ since the semi-perimeter of its cluster of pluses is strictly greater than $2l_c-1$ and the number of zeros $k$ is smaller than $2l_c-1$. This means that either $\alpha_2 \geq 1$ or $\alpha_3 \geq 1$. 
In particular, if $\alpha_3=0$ (and $\alpha_1 \geq 1$) then $\eta$ contains a single cluster of pluses with a shape different from a quasi-square and in this case $\alpha_1+\alpha_2 >2l_c-1$. Hence, by using \eqref{eq:energy_eta_no_quasi_quad} and \eqref{eq:energy_critical_conf}, we obtain
\begin{align}\label{eq:energy_eta_no_quasi_quad_diff}
   H(\eta)-H(\sigma_c) & > 2J(\alpha_2+\alpha_3) +(n^-_{\sigma_c}-n^-_{\eta})(\lambda-h) \notag \\
   & \geq 2J+(n^0_{\eta}-n^0_{\sigma_c})(\lambda-h) \notag \\
   &= 2J+(k-2l_c+1)(\lambda-h) \notag \\
   &\geq 2J+(2-2l_c)(\lambda-h) \geq 0
\end{align}
where the last inequality follows by \eqref{eq:critical_length}, $k\geq 1$ and $J>>\lambda>h>\frac{\lambda}{2}$.

Otherwise, if $\alpha_3 \geq 1$ we note that $\eta$ contains at least two disconnected clusters of pluses and for the geometry of the lattice, also in this case we have $\alpha_1+\alpha_2 >2l_c-1$. Thus, arguing as above, we obtain the same result as in \eqref{eq:energy_eta_no_quasi_quad_diff}.
\end{proof}

%
\begin{proof}[Proof of Lemma \ref{lem:min_crit_quad_no_boundary}]
Let $\eta \in \text{argmax}_{\xi \in \mathcal{M}_{n^+_c}}H(\xi)$. By Lemma \ref{lem:cluster_pluses_different_quasi_square} we have that $\eta$ contains a single quasi-square $Q$ of pluses. Moreover, we may apply Lemma \ref{lem:region_zero_minus_1} or Lemma \ref{lem:region_zero_minus_2} in the region $\Lambda \setminus Q$, then we have $\eta_{\Lambda \setminus (Q \cup \partial^+ Q)}=-1_{\Lambda \setminus (Q \cup \partial^+ Q)}$, otherwise $\eta \not \in M$. 
\end{proof}

\begin{proof}[Proof of Lemma \ref{lem:sigma_c_minimum}]
Let $\eta \in \mathcal{M}_{n^+_c}$. By Lemma \ref{lem:min_crit_quad_no_boundary}, we have that $\eta_Q=+1_Q$ and $\eta_{{\Lambda \setminus (Q \cap \partial^+ Q)}}=-1_{\Lambda \setminus (Q \cap \partial^+ Q)}$, otherwise $\eta \not \in M$.
We will prove that if $\eta \in M$, then $\eta_{\partial^+Q}=0_{\partial^+Q}$. Suppose that there exists $x,y \in \partial^+ Q \cap \Lambda$, $|x-y|=1$, such that $\eta(x)=0$ and $\eta(y)=-1$, then we find $\eta' \in \mathcal{M}_{n^+_c}$ with $H(\eta')<H(\eta)$ by applying Table \ref{tab:heu000} at row 5 (by flipping the minus in $y$ in to zero). Then, we have either $\eta(x)=-1$ for all $x\in \partial^+ Q \cap \Lambda$, or $\eta(x)=0$ for all $x\in \partial^+ Q \cap \Lambda$. However, in the first case there exists $x,y \in \partial^+ Q \cap \Lambda$, $|x-y|=1$, such that $\eta(x)=\eta(y)=-1$, then by flipping the two minuses in $x$ and $y$ into zero, we obtain a configuration $\eta' \in \mathcal{M}_{n^+_c}$ with $H(\eta')=H(\eta)-2J+2(\lambda-h)<H(\eta)$ by applying Table \ref{tab:heu000} at rows 3 and 5. Thus, $\eta_{\partial^+Q \cap \Lambda}=0_{\partial^+Q \cap \Lambda}$ otherwise $\eta \not \in M$. \\
According to the zero-boundary conditions, the energy of $\eta$ depends on the position of $Q$, then $\eta$ can contain either a frame, or a chopped boundary frame or a chopped corner frame. The energy of these three configuration is computed in \eqref{eq:energy_frame} and by using \eqref{eq:difference_energy_frame}, we can conclude that $\sigma_c$ is the unique $\text{argmax}_{\xi \in \mathcal{M}_{n^+_c}}H(\xi)$.
\end{proof}

\begin{proof}[Proof of Lemma \ref{lem:path_sigma_c}]
Consider a path $\underline{\omega}$ as in the assumption, and we suppose by contradiction that $\Phi(\underline{\omega})<H(\tilde \sigma_c)$. If there exists $\eta \in \mathcal{M}_{n^+_c+1}\cap \underline{\omega}$ such that $\sigma_c \sim \eta$, then by Table \ref{tab:heu000}, we obtain  
\begin{equation}
    \Phi(\underline{\omega})\geq H(\eta)=H(\sigma_c)+4J-(\lambda+h)>H(\tilde\sigma_c).
\end{equation}
Then, we have $\sigma_c \sim \eta$ where $\eta \in \mathcal{M}_{n^+_c}\cap \underline{\omega}$. 
We note that, according to Table \ref{tab:heu000}, if the configuration contains a cluster of pluses with a quasi-square shape, then the minimal energy contribution to add a plus is $2J-(\lambda+h)$. This is possible by flipping a zero with \,\, 
{\footnotesize
$
\begin{matrix}
&0&\\[-0.3em]
+&\cdot&0 \\[-0.3em]
\, &0& \,
\end{matrix}
$ 
}
\,\, nearest neighbors into a plus, however all zero spins in $\sigma_c$ have only a plus and at most two zeros nearest neighbors. Thus, we obtain $\eta$ from $\sigma_c$ by flipping a minus into zero in order to create a zero with \,\, 
{\footnotesize
$
\begin{matrix}
\, & 0 & \, \\[-0.3em]
+ & \cdot & 0 \\[-0.3em]
\, & 0 & \,
\end{matrix}
$
}
\,\,  nearest neighbors, otherwise $\Phi(\underline{\omega})\geq H(\tilde\sigma_c)$. We note that the minimal energy contribution to obtain $\eta$ is $(\lambda-h)$ when we flip a minus with at most two minuses nearest neighbors to zero, see Table \ref{tab:heu000}. Then, we obtain from $\eta$ a configuration $\eta' \in \mathcal{M}_{n^+_c}$ by adding a plus in the site where there is the zero with \,\, 
{\footnotesize
$
\begin{matrix}
\, & 0 & \, \\[-0.3em]
+ & \cdot & 0 \\[-0.3em]
\, & 0 & \,
\end{matrix}
$ 
}
\,\, nearest neighbors.   
However, also in this case we have $\Phi(\underline{\omega})\geq H(\eta')=H(\sigma_c)+(\lambda-h)+2J-(\lambda+h)=H(\tilde \sigma_c)$ and this is a contradiction.
We observe that if we remove more than one minus before adding a plus, the communication height is even greater. 
\end{proof}

\begin{proof}[Proof of Lemma \ref{lem:step_1_minmax}]
Let $\underline{\omega}=(\eta_1,....\eta_n)$, $n \in \mathbb{N}$, be a path as in the assumption. We consider the part of 
$\underline{\omega}$ from $\mathcal{M}_{n^+_c}$ to $\mathcal{M}_{n^+_c+1}$ and let $\eta_i \in \mathcal{M}_{n^+_c} \cap \underline{\omega}$ such that $\eta_i \sim \eta_{i+1} \in \mathcal{M}_{n^+_c +1}$. If $\eta_i \equiv \sigma_c$, then we conclude by applying Lemma \ref{lem:path_sigma_c}. Thus, suppose that $\eta_i \not \equiv \sigma_c$ and assume by contradiction that $\Phi(\underline{\omega})<H(\tilde \sigma_c)$. We have $H(\eta_{i})>H(\sigma_c)$ by Lemma \ref{lem:sigma_c_minimum} and in particular since $\eta_i \not \equiv \sigma_c$ from the Table \ref{tab:heu000} we have $H(\eta_i)=H(\sigma_c)+2Ja+b(\lambda-h)$ with $a\in \mathbb{N}$ and $b\in \mathbb{Z}$ such that $2Ja+b(\lambda-h)>0$. Moreover, $\eta_i \in \mathcal{M}_{n^+_c}$, $\tilde\sigma_c \in \mathcal{M}_{n^+_c+1}$ and $H(\eta_i)<H(\tilde\sigma_c)$, in particular $H(\eta_i) \leq H(\tilde\sigma_c)-2J+(\lambda+h)$ according to Table \ref{tab:heu000}. Then, by \eqref{eq:energy_sigma_protuberance} we have 
\begin{align}
   H(\sigma_c)&<H(\sigma_c)+2Ja+b(\lambda-h) \notag \\
   &=H(\eta_i) \leq H(\tilde\sigma_c)-2J+(\lambda+h) \notag \\
   &=H(\sigma_c)+(\lambda-h).
\end{align}
Follows that $a=0$, $b=1$ and so 
\begin{equation}\label{eq:eta_i_sigma_c}
    H(\eta_i)=H(\sigma_c)+(\lambda-h).
\end{equation}
This also implies that $\eta_i$ contains one less minus spin with respect to $\sigma_c$.
Moreover, $\eta_i \sim \eta_{i+1}$ then $\eta_{i}$ and $\eta_{i+1}$ differs for only one plus. Hence, let $\alpha \in \mathbb{Z}$, by using Table \ref{tab:heu000}, \eqref{eq:eta_i_sigma_c} and \eqref{eq:energy_sigma_protuberance}, we obtain
\begin{align}
    H(\tilde \sigma_c) & >\Phi(\underline{\omega})\geq H(\eta_{i+1}) \notag \\
    &=H(\eta_{i})+2J \alpha-(\lambda+h) \notag \\
    &=H(\sigma_c)+(\lambda-h)+2J \alpha-(\lambda+h) \notag \\
    &=H(\tilde \sigma_c)+2J(\alpha-1).
\end{align}
 Thus $\alpha \leq 0$. This implies that $H(\eta_{i+1}) \leq H(\eta_{i})-(\lambda+h)$ and, according to Table \ref{tab:heu000}, we have that $\eta_i$ contains (at least) a zero spin with at least two pluses nearest neighbors, that is  
 \begin{center}
{\footnotesize
$
\begin{matrix}
\, &  0 & \, \\[-0.3em]
+ & \cdot & 0 \\[-0.3em]
\, & + & \,
\end{matrix}
$
}
, \,\, 
{\footnotesize
$
\begin{matrix}
\, &  0 & \, \\[-0.3em]
+ & \cdot & + \\[-0.3em]
\, & + & \,
\end{matrix}
$
}
, \,\,
{\footnotesize
$
\begin{matrix}
\, &  - & \, \\[-0.3em]
+ & \cdot & + \\[-0.3em]
\, & + & \,
\end{matrix}
$
}
, \,\,
{\footnotesize
$
\begin{matrix}
\, &  + & \, \\
+ & \cdot & + \\
\, & + & \,
\end{matrix}
$
}. 
\end{center}
This implies that the cluster of pluses in $\eta_i$ has at least a convex corner and so it has a shape different from a quasi-square. For this reason and the fact that $\eta_i$ contains one less minus spin with respect to $\sigma_c$, we proceed as in proof of Lemma \ref{lem:cluster_pluses_different_quasi_square} and we can apply \eqref{eq:eta_no_quasi_quad_sigma_c}. Thus, we obtain
 \begin{align}
H(\eta_i)>H(\sigma_c)+(\lambda-h)
 \end{align}
 and this is a contradiction for \eqref{eq:eta_i_sigma_c}.
\end{proof}

\begin{proof} [Proof of Lemma \ref{lem:path_tilde_sigma_c}]
   Let $\underline{\omega}$ be a path from $\tilde \sigma_c \in \mathcal{M}_{n^+_c+1}$ to $\mathcal{M}_{n^+_c+2}$ without loops. Then, if there exists $\eta\sim \tilde \sigma_c$ such that $\eta\in\underline{\omega}\cap\mathcal{M}_{n^+_c+2}$ then by Table \ref{tab:heu000}, we have $\Phi(\underline{\omega}) \geq H(\tilde \sigma_c)+2J-(\lambda+h)$ and we conclude by \eqref{eq:energy_sigma_protuberance} and \eqref{eq:energy_saddle}. Otherwise, if along $\underline{\omega}$ we have that $\tilde\sigma_c \sim \eta$ with $\eta \in \mathcal{M}_{n^+_c+1}$ then we conclude $\Phi(\underline{\omega}) \geq H(\tilde \sigma_c)+(\lambda-h)$ by using Table \ref{tab:heu000}, \eqref{eq:energy_sigma_protuberance} and \eqref{eq:energy_saddle}.
\end{proof}

\begin{proof} [Proof of Lemma \ref{lem:equivalence_tilde_sigma_c}]
    Let $\eta \in \mathcal{M}_{n^+_c+1}$ be a configuration such that $H(\eta)=H(\tilde \sigma_c)$. Then, by using the same partition of the columns and rows in the proof of Lemma \ref{lem:cluster_pluses_different_quasi_square} (see conditions a.,b.,c. and Figure \ref{fig:RS}, we have
    \begin{align}
        &H(\eta)-H(\tilde \sigma_c)=4J(\alpha_1-2l_c)+2J(\alpha_4-2L+2l_c) \notag \\
        &+6J(\alpha_2+\alpha_3)+(n^-_c-n^-_\eta)(\lambda-h) \notag \\
        &=2J(2\alpha_1+\alpha_4+3\alpha_2+3\alpha_3-2L-2l_c) \notag \\
        &+(n^-_c-n^-_\eta)(\lambda-h).
    \end{align}
    Since $H(\eta)=H(\tilde \sigma_c)$, follows that
    \begin{align}
      \begin{cases}
        &2\alpha_1+\alpha_4+3(\alpha_2+\alpha_3)=2(L+l_c) \notag \\
        &n^-_c=n^-_\eta
        \end{cases}  
    \end{align}
    The second equality implies that the number of minuses in $\eta$ is the same as in $\tilde \sigma_c$. Recalling that $\sum_{i=1}^4 \alpha_i=2L$, we obtain
    $\alpha_1+2(\alpha_2+\alpha_3)=2l_c$. Moreover, $\sum_{i=1}^3 \alpha_i \geq 2l_c$ indeed the minimal semi-perimeter of a cluster of pluses with area $n^+_c$ is $2l_c$ by \cite{alonso1996three}. Then, we derives $\alpha_2+\alpha_3=0$. From the definitions of $\alpha_2$ and $\alpha_3$, it follows that the bonds of type $(+,-)$ are not present and that the number of bonds $(+,0)$ in each column and row is at most two. Hence, $\alpha_1=2l_c$ and $\alpha_4=2(L-l_c)$. This implies that the union of the clusters of pluses has a semi-perimeter equal to $2l_c$. 
    By \cite{alonso1996three}, such union of clusters has to be contained either in a square with both side length $l_c$ or in a rectangle with side length $l_c-1$ and $l_c+1$. Moreover, the cluster of pluses can not contain more than one protuberance along each side (otherwise $\alpha_3 \neq 0$), see Figure \ref{fig:examples_cluster_pluses_in_square}. We note that the cluster is in the corner of $\Lambda$, otherwise either the number of zeros is greater than the number of zeros in $\tilde \sigma_c$, or $\eta$ contains at least a bond $(+,-)$. Finally, we observe that in case the configuration contains more than one strip of minuses in at least a column or a row of $\Lambda$, then $\eta$ contains a column among the $\alpha_4$ columns with an energy contribution of $4J$ instead of $2J$ and so $H(\eta) > H(\tilde \sigma_c)$. For the same reason, the protuberance attached to the cluster is at distance one from the boundary of $\Lambda$, see Figure \ref{fig:examples_cluster_pluses_in_square}. We may conclude that $\eta \in \mathscr{S}$.     
    \begin{figure}[H]
\begin{center}
    \includegraphics[scale=0.24]{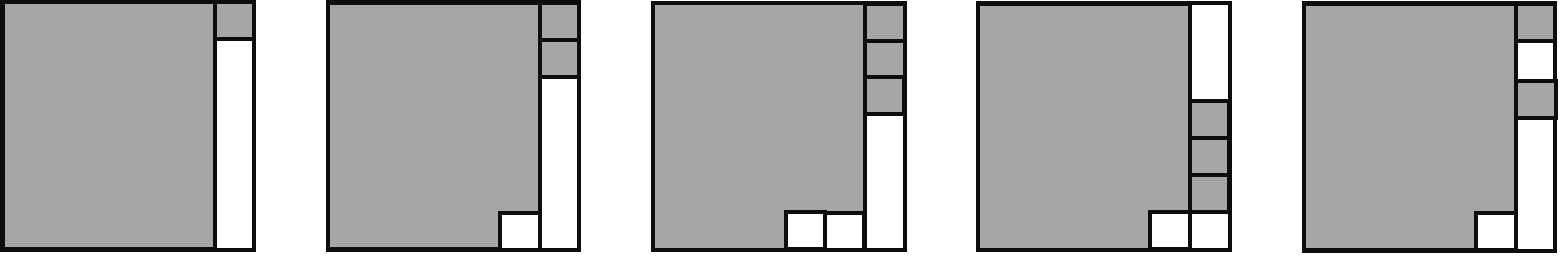}
    \caption{ Examples of clusters of pluses with area $n^+_c$ and the same perimeter of a square with side length $l_c$. The first three clusters are contained in a configuration of $\mathscr{S}$.}\label{fig:examples_cluster_pluses_in_square}
\end{center}
    \end{figure}
\end{proof}

\begin{proof}[Proof of Lemma \ref{lem:conf_S_different_tilde_sigmac}]
We consider a configuration $\eta \in \mathscr{S} \setminus \{ \tilde \sigma_c\}$. By the property of $\mathscr{S}$, the cluster of pluses of every configuration in $\mathscr{S}$ has perimeter equal to $2l_c$, but the number of the sites along the external-boundary of this cluster (the so called \emph{site-perimeter}) changes according to the number of its concave angles. We note that $\tilde \sigma_c$ contains $2l_c$ sites along the external-boundary of the cluster of pluses, instead the other configurations contains $k>1$ concave angles (see the second and the third pictures in Figure \ref{fig:examples_cluster_pluses_in_square} for two examples with $k=2$). Moreover, every configuration in $\mathscr{S}$ contains $2l_c$ zero spins, thus $\eta \not \equiv \tilde \sigma_c$ contains $k \geq 1$ zero spins at distance strictly greater than one from the cluster. In particular, this zero spins are in the minimal rectangle containing the cluster of pluses and they are attached from the zero spins at distance one from the cluster in order not to create more than one strip of minuses in each column and row of $\Lambda$. \\
We consider the configuration $\eta_k$ obtained from $\eta$ by flipping these $k$ zero spins to minus, see Figure \ref{fig:conf_S_path}. For the definition of the set $\mathcal{S}$ and by using Table \ref{tab:heu000}, we have $H(\eta_k)=H(\eta)-k(\lambda-h)$. By Table \ref{tab:heu000}, we note that $\eta_k$ is a local minimum, indeed every path from it is an up-hill path. In particular, since to reach $\muno$, $\underline{\omega}$ must cross all the manifold $\mathcal{M}_{n^+_c+1}$, $\mathcal{M}_{n^+_c}$, $\mathcal{M}_{n^+_c-1}$, ..., $\mathcal{M}_{0}$, then we consider $\eta_m$ be the configuration with a protuberance with cardinality one obtained from $\eta_k$, see Figure \ref{fig:conf_S_path} for an example of $\eta_m$. The path that connected $\eta_k$ with $\eta_m$ is a two-steps down-hill path, where the up-hill is the minimal positive energy quantum $\lambda-h$, thus $\Phi(\eta,\eta_m) \geq \Phi(\underline{\omega})$. We obtain $H(\eta_m) = H(\eta_k)+m(\lambda+h)-(m-1)(\lambda-h)$ by flipping $m$ pluses to zero and $m-1$ zeros to minus as in Figure \ref{fig:conf_S_path} and by using Table \ref{tab:heu000}. We note that $m \geq k$, indeed the cluster of pluses contains $n^+_c+1$ pluses in a rectangle $l_c \times l_c$ (or $(l_c-1) \times (l_c+1)$) then when it contains $\gamma$ concave angles then the protuberance has cardinality more than $\gamma$ because the pluses that are not present in the sites of the concave angles must be located along the protuberance to be inside the rectangle. 

Assume by contradiction that $\Phi(\underline{\omega}) \leq H(\sigma_s)$. Thus, we have

\begin{align}
    H(\sigma_s) & \geq H(\eta_m) \geq H(\eta_k)+m(\lambda+h)-(m-1)(\lambda-h) \notag \\
    &\geq H(\eta)-(k+m-1)(\lambda-h)+m(\lambda+h) \notag \\
    &= H(\tilde\sigma_c)-(k+m-1)(\lambda-h)+m(\lambda+h) \notag \\
    & \geq H(\tilde\sigma_c)+(1-k)\lambda+(2m+k-1)h \notag \\
    & = H(\sigma_s)-k\lambda+(2m+k)h \notag \\
    & \geq H(\sigma_s)-k\lambda+3kh
\end{align}
where the equality is obtained by \eqref{eq:energy_sigma_protuberance} and \eqref{eq:energy_saddle} and the last inequality follows from $m \geq k$. Then, we obtain a contradiction, indeed $H(\sigma_s) \geq H(\sigma_s)-k\lambda+3kh \geq H(\sigma_s)-\lambda+3h>H(\sigma_s)$, since $k \geq 1$ and $h>\lambda/2$.
\begin{figure}[H]
\begin{center}
    \includegraphics[scale=0.5]{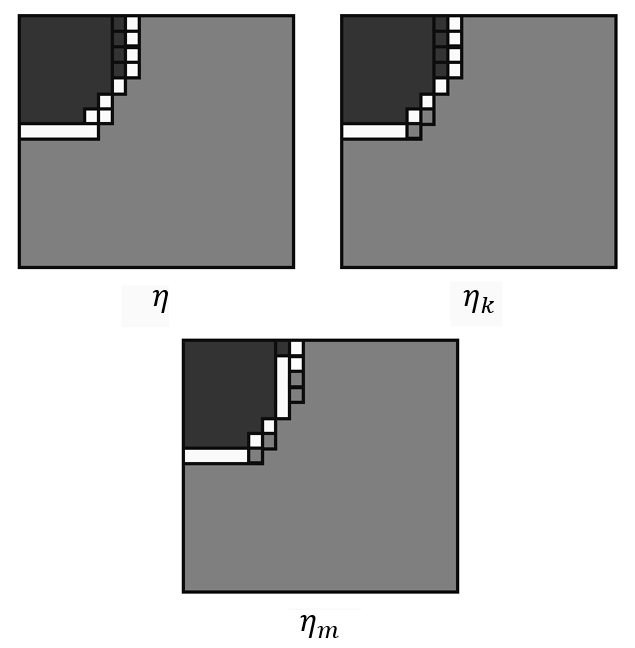}
    \caption{On the left, a configuration $\eta \in \mathscr{S}$. In the middle $\eta_k$, with $k=2$ in this case. On the right, an example of the configuration $\eta_m$ with $m=3$ and such that $H(\eta_m)=H(\eta_k)+m(\lambda+h)-(m-1)(\lambda-h)$.}\label{fig:conf_S_path}
\end{center}
    \end{figure}
\end{proof}

\begin{proof}[Proof of Lemma \ref{lem:step_2_minmax}]
Let $\underline{\omega}=(\eta_1,...,\eta_k,...,\eta_n)$, $n \in \mathbb{N}$, be a path as in the assumption where $\eta_1 \in \mathcal{M}_{n^+_c}$ and $\eta_n \in \mathcal{M}_{n^+_c+2}$. Thus, along $\underline{\omega}$, there exists a configuration $\eta_k \in \mathcal{M}_{n^+_c+1}$. Let $\underline{\omega}_i$ and $\underline{\omega}_f$ be respectively the initial and the final part of the path $\underline{\omega}$, i.e., $\underline{\omega}_i=(\eta_1,...,\eta_k)$ and $\underline{\omega}_f=(\eta_{k+1},...,\eta_n)$. By Lemma \ref{lem:step_1_minmax}, we have $\Phi(\underline{\omega}_i) \geq H(\tilde \sigma_c)$. This implies that $\Phi(\underline{\omega}) \geq \Phi(\underline{\omega}_i) \geq H(\tilde \sigma_c)$. 
Assume by contradiction $\Phi(\underline{\omega}) < H(\sigma_s)$, then
\begin{equation}
    H(\sigma_s) > \Phi(\underline{\omega}) \geq H(\tilde \sigma_c)=H(\sigma_s)-(\lambda-h),
\end{equation}
where the last equality follows by \eqref{eq:energy_sigma_protuberance} and \eqref{eq:energy_saddle}. Thus, $\Phi(\underline{\omega})=\Phi(\underline{\omega}_i)=H(\tilde \sigma_c)$ because $\lambda-h$ is the minimal energy difference of the model. Follows that there exists a configuration $\tilde \eta$ along $\underline{\omega}_i$ such that $H(\tilde \eta)=H(\tilde \sigma_c)$ and $\tilde \eta \in \mathcal{M}_{n^+_c+1}$. Thus, by Lemma \ref{lem:equivalence_tilde_sigma_c}, we obtain $\tilde \eta \in \mathscr{S}$. If $\tilde \eta \neq \tilde \sigma_c$, then by Lemma \ref{lem:conf_S_different_tilde_sigmac} we obtain $\Phi(\underline{\omega}) \geq \Phi(\underline{\omega}_i) >H(\sigma_s)$ and we conclude the proof. Otherwise $\tilde \eta \equiv \tilde \sigma_c$ and we conclude by applying Lemma \ref{lem:path_tilde_sigma_c}.
\end{proof}

\par\noindent

\appendix
\renewcommand{\thesubsection}{\Alph{section}.\arabic{subsection}}
\renewcommand{\theequation}{\Alph{section}.\arabic{equation}}
\renewcommand{\thefigure}{\Alph{section}.\arabic{figure}}

\begin{acknowledgments}
ENMC acknowledges that this work has been done under the 
framework of GNMF.
\end{acknowledgments}
\end{multicols}
\bibliographystyle{abbrv}
\bibliography{cjs-bcg_zero}
\end{document}